\documentclass[UKenglish,cleveref, autoref, thm-restate, algorithmic, algorithm]{lipics-v2021}
\usepackage{booktabs}
\usepackage{algorithm}
\usepackage{algorithmic}
\usepackage{todonotes}
\usepackage{changepage}
\usepackage{cite}
\usepackage{comment}

\usepackage{moresize}

\graphicspath{{./figures/}}

\title{Instance-Optimal Imprecise Convex Hull}

\author{Sarita de Berg}{IT University of Copenhagen, Denmark}{debe@itu.dk}{0000-0001-5555-966X}{}

\author{Ivor van der Hoog}{IT University of Copenhagen, Denmark}{ivva@itu.dk}{https://orcid.org/0009-0006-2624-0231}{This project has received funding from the European Union's Horizon 2020 research and innovation programme under the Marie Sk\l{}odowska-Curie grant agreement No 899987}

\author{Eva Rotenberg}{IT University of Copenhagen, Denmark}{erot@itu.dk}{0000-0001-5853-7909 }{}{}

\author{Daniel Rutschmann}{IT University Copenhagen, Denmark}{daru@dtu.dk}{ https://orcid.org/0009-0005-6838-2628}{}

\author{Sampson Wong}{University of Copenhagen, Denmark}{sawo@di.ku.dk}{https://orcid.org/0000-0003-3803-3804}{This project has recieved funding from the European Union's Sk\l{}odowska-Curie Actions Postdoctoral Fellowship grant agreement No 101146276.}

\funding{This research is supported by the Carlsberg Foundation Fellowship CF21-0302 ``Graph Algorithms with Geometric Applications'' and the VILLUM Foundation grants VIL37507 ``VIL37507 Recomputations for Changeful Problems''.}

\nolinenumbers
\hideLIPIcs

\authorrunning{Sarita de Berg, Ivor van der Hoog, Eva Rotenberg, Daniel Rutschmann, Sampson Wong}

\EventEditors{Anne Benoit, Haim Kaplan, Sebastian Wild, and Grzegorz Herman}
\EventNoEds{4}
\EventLongTitle{33rd Annual European Symposium on Algorithms (ESA 2025)}
\EventShortTitle{ESA 2025}
\EventAcronym{ESA}
\EventYear{2025}
\EventDate{September 15--17, 2025}
\EventLocation{Warsaw, Poland}
\EventLogo{}
\SeriesVolume{351}
\ArticleNo{23}

\Copyright{Sarita de Berg, Ivor van der Hoog, Eva Rotenberg, Daniel Rutschmann, Sampson Wong} 

\ccsdesc[100]{Theory of computation~Computational Geometry} 

\keywords{convex hull, imprecise geometry preprocessing model, partial information}  

\usepackage{xspace}

\graphicspath{ {figures/} }

\newcommand {\bE} {\mathbb {E}}

\newcommand{\retrieve}{ \hspace{-0.05cm }\ensuremath{\odot}_{\hspace{-0.05cm}\textit{\ssmall P}} \hspace{-0.03cm}}
\newcommand{\rfp}{r(F,P)}

\DeclareMathOperator{\CH}{CH}
\DeclareMathOperator{\OCH}{OCH}
\DeclareMathOperator{\MCD}{MCD}
\DeclareMathOperator{\PHT}{PHT}

\DeclareMathOperator{\poly}{poly}

\DeclareMathOperator{\band}{band}

\DeclareMathOperator{\precorder}{\preceq}

\begin{document}


\maketitle

\begin{abstract}
Imprecise measurements of a point set $P = (p_1, \ldots, p_n)$ can be modelled by a family of regions $F = (R_1, \ldots, R_n)$, where each imprecise region $R_i \in F$ contains a unique point $p_i \in P$. A \emph{retrieval} models an accurate measurement by replacing an imprecise region~$R_i$ with its corresponding point~$p_i$. 

We construct the convex hull of an imprecise point set in the plane, by determining the cyclic ordering of the convex hull vertices of $P$ as efficiently as possible. Efficiency is interpreted in two ways: (i) minimising the number of retrievals, and (ii) the computation time to determine the set of regions that must be retrieved. 

Previous works focused on only one of these two aspects: either minimising retrievals or optimising algorithmic runtime. Our contribution is the first to simultaneously achieve both. Let $r(F, P)$ denote the minimal number of retrievals required by any algorithm to determine the convex hull of $P$ for a given instance $(F, P)$. For a family $F$ of $n$ constant-complexity polygons, our main result is a reconstruction algorithm that performs $\Theta(r(F, P))$ retrievals in $O(r(F, P) \log^3 n)$ time. 

Compared to previous approaches that achieve optimal retrieval counts, we improve the runtime per retrieval from polynomial to polylogarithmic. 
We extend the generality of previous results to simple $k$-gons, to pairwise disjoint disks with radii in $[1,k]$, and to unit disks where at most $k$ disks overlap in a single point. Our runtime scales linearly with $k$.
\end{abstract}

\newpage
\section{Introduction}
Imprecision is inherent in real-world data. It arises from rounding errors in floating-point computations, inaccuracies in measurement, and delayed sampling in GPS devices. In many scenarios, imprecise data can be refined at a cost by computing exact values or by taking additional samples. This is formalised in the model of \emph{imprecise geometry} introduced by Held and Mitchell~\cite{held2008triangulating}, and studied further in~\cite{devillers2011delaunay,held2008triangulating,van2010preprocessing,evans2011possible, ezra2013convex,loffler2010delaunay,loffler2013unions}.  

\subparagraph{Imprecise geometry.} An imprecise point set~\cite{held2008triangulating} is defined as a family of regions $F = (R_1, \ldots, R_n)$, where each region $R_i$ contains a unique but unknown point $p_i$. A \emph{realisation} $P \sim F$ is a sequence $P = (p_1, \ldots, p_n)$ where $p_i \in R_i$. An input instance is a pair $(F, P)$. A \emph{retrieval} operation reveals the precise location of a point $p_i$, replacing $R_i$ with $p_i$. 

Let $F \retrieve B$ denote the family~$F$ after retrieving all $R_i \in B$, where $B \subset F$. 
The aim of a \emph{reconstruction algorithm} is to identify a subset $B \subset F$ such that for all realisations $P_1, P_2 \sim F \retrieve B$, the algorithm’s output (e.g., the cyclic order of the
region indices of the points around the convex hull) is identical for $P_1$ and~$P_2$. Figure~\ref{fig:problem} illustrates an example. After retrieving the subset $B = \{ R_3,R_4\}$, the cycling ordering of the convex hull vertices is identically~$(p_1,p_3,p_2,p_4)$ for all realisations of $F \retrieve B$.


We evaluate reconstruction algorithms by three criteria: the preprocessing time, the total number of retrievals, and the running time per retrieval. Next, we consider instance-optimal algorithms, which minimise the total number of retrievals in the strictest sense.

\begin{figure}[b]
    \centering
    \includegraphics[width = \linewidth]{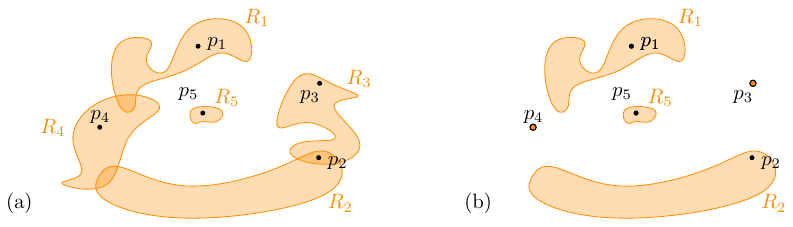}
    \caption{ (a) A family of regions $F = R_1,\ldots,R_5$ and a sequence $P = (p_1,\ldots,p_5)$ with $P \sim F$. 
    (b) If we retrieve $R_3$ and $R_4$ to obtain $F'$ then for all $P' \sim F'$ the convex hull equals $(p_1, p_3, p_2, p_4)$. Note that if $p_3$ would lie in $R_2 \cap R_3$ instead, then retrieving only $R_3$ and $R_4$ does not suffice.     
    }
    \label{fig:problem}
\end{figure}

\subparagraph{Worst-case results.}
Many geometric problems have been studied in imprecise geometry, including Delaunay triangulations~\cite{devillers2011delaunay,held2008triangulating,loffler2010delaunay,van2010preprocessing}, convex hulls~\cite{evans2011possible, ezra2013convex}, Gabriel graphs~\cite{loffler2010delaunay}, and onion decompositions~\cite{loffler2013unions}. 
For these geometric problems, it is known that any reconstruction algorithm must, in the worst case, use $\Omega(n)$ retrievals.
Most of the existing algorithms assume that the regions in $F$ are pairwise disjoint unit disks~\cite{devillers2011delaunay,held2008triangulating,loffler2010delaunay,van2010preprocessing, evans2011possible, loffler2013unions}. Under these assumptions, deterministic worst-case optimal algorithms exist with $O(n \log n)$ preprocessing time, $O(n)$ retrievals, and $O(n)$ reconstruction time.
Ezra and Mulzer~\cite{ezra2013convex} allow $F$ to consist of arbitrary lines and reconstruct  the convex hull using $O(n)$ retrievals and $O(n \cdot \alpha(n))$ expected time where $\alpha(n)$ denotes the inverse Ackermann function. 
Buchin, L\"{o}ffler, Morin and Mulzer~\cite{buchin2011preprocessing} permit overlapping unit disks of \emph{ply} $k$ (i.e., at most $k$ disks intersect each point in the plane). 
We summarize these results in the top three rows of Table~\ref{tab:my_label}. 

\subparagraph{Instance-optimality.}

In many instances, worst case bounds are overly pessimistic. Intuitively, an algorithm is instance-optimal if, for every input $(F, P)$, no other algorithm performs better on that instance.
Constructing instance-optimal algorithms is challenging, as they must match every alternative algorithm, including those tailored for specific instances. 

Afshani, Barbay, and Chan~\cite{afshani2017instance} proved that, for many geometric problems including constructing the convex hull, instance-optimal reconstruction time is unachievable. Likewise, it is easy to see that instance-optimality in the number of retrievals cannot be guaranteed in general:
E.g., let $F$ consist of two overlapping intervals $R_1$ and $R_2$ in $\mathbb{R}$, with $R_1 \setminus R_2 \neq \emptyset$ and $R_2 \setminus R_1 \neq \emptyset$. Define $P_1$ with $p_1 \in R_1 \setminus R_2$ and $p_2 \in R_1 \cap R_2$, and $P_2$ with $p_1 \in R_1 \cap R_2$ and $p_2 \in R_2 \setminus R_1$. In $(F, P_1)$, retrieving $R_1$ suffices; in $(F, P_2)$, retrieving $R_2$ suffices. No algorithm can distinguish between the two instances without making a retrieval, and hence must use two retrievals on one of the instances. Since exact instance-optimality cannot be obtained, we instead aim for \emph{asymptotic} instance-optimality. 

Formally, we denote for any input instance by $r(F, P)$ the minimum integer such that there exists a subset $B \subset F$ of size $r(F,P)$ where all realizations $P' \sim (F \retrieve B)$ have the same vertex-ordering along the convex hull of $P'$. 
We note that $r(F, P)$ is, equivalently, the optimal number of retrievals.
We say that an algorithm is \emph{instance-optimal} if for all inputs $(F, P)$ it uses $\Theta(r(F, P))$ retrievals. 
We highlight that $r(F, P)$ depends on both the region family $F$ and the specific realisation $P$.
To illustrate, suppose $F$ consists of nested rectangles $R_n \subset R_{n-1} \subset \ldots \subset R_1$. If $P \sim F$ is such that the convex hull of $(p_1, p_2, p_3, p_4)$ contains all of $R_5$, then retrieving just those four suffices: $r(F, P) = 4$. But if all $p_i$ coincide then $r(F, P) = n$.

\subparagraph{Prior instance-optimal work.}
Only a few instance-optimal reconstruction algorithms are known (see Table~\ref{tab:my_label}).  Bruce, Hoffmann, Krizanc, and Raman~\cite{bruce2005efficient} presented the first such algorithm for the convex hull.
Their approach is computationally expensive, using an unspecified but superlinear time per retrieval (we discuss this in the appendix). Subsequent works show that polylogarithmic time per retrieval is achievable, albeit on geometric structures that are much simpler than the convex hull.

Van der Hoog, Kostitsyna, L\"offler, and Speckmann~\cite{van2019preprocessing} introduced the first near-linear instance-optimal algorithm, solving the sorting problem for one-dimensional intervals. Their algorithm preprocesses the intervals in $O(n \log n)$ time and uses $\Theta(r(F, P))$ retrievals, with at most logarithmic time per retrieval.
Later, Van der Hoog, Kostitsyna, L\"offler, and Speckmann~\cite{van2022preprocessing} extended this approach to two-dimensional inputs for Pareto front reconstruction, under the assumption that $F$ consists of pairwise disjoint axis-aligned rectangles. Their method also guarantees at most logarithmic time per retrieval. 

Here, we show that polylogarithmic time per retrieval is achievable for reconstructing the convex hull with an instance-optimal number of retrievals. Moreover, our work generalises previous works in that we can handle overlapping regions in two dimensions, whereas previous works could only support one-dimensional inputs or disjoint two dimensional inputs. 

\subparagraph{Simultaneous work.}
Simultaneously and independently of this paper, L\"{o}ffler and Raichel~\cite{loffler2025preprocessing} developed an algorithm for reconstructing the convex hull when $F$ is a set of unit disks of ply~$k$. Their number of retrievals (Table~\ref{tab:my_label}) is not instance-optimal but rather worst-case optimal for every instance $F$. 
Formally, they use $O(w(F))$ retrievals for some region-dependent value $w(F)$ with $r(F, P) \leq w(F) \leq n$. 
We discuss their result in the appendix.

\subparagraph{Instance optimality in other fields.}
The study of instance-optimal algorithms extends beyond computational geometry. A prime example is sorting under partial information. Given a partial order $O$ over a set $X$ and an unknown linear extension $L$, the goal is to sort $X$ using a minimum number of comparisons, where a comparison queries the order of a pair $(x, y) \in X$ under $L$. This is directly analogous to retrievals.
Kahn and Saks~\cite{kahn_balancing_1984} introduced an exponential-time algorithm that is instance-optimal in the number of comparisons. Since then, significant progress has been made on improving the runtime of such algorithms~\cite{kahn_entropy_1992,cardinal_sorting_2013}, culminating in near-linear time algorithms that are instance-optimal in both runtime and comparisons~\cite{van2024tight, van2025simpler, Haeupler25}.
Another example is the bidirectional shortest path problem. Haeupler, Hlad{\'\i}k, Rozho{\v{n}}, Tarjan, and T{\v{e}}tek~\cite{haeupler2025bidirectional} show an algorithm that finds the shortest path between two nodes $s$ and $t$ using an instance-optimal number of edge-weight comparisons.

\subparagraph{Our contributions.}
We present the first algorithm for convex hull reconstruction that is instance-optimal while only requiring polylogarithmic time per retrieval. If~$F$ is a family of~$n$ constant-complexity simple polygons, we preprocess $F$ in $O(n \log^3 n)$ time and reconstruct the convex hull of any $P \sim F$ using $\Theta(r(F, P))$ retrievals and $O(r(F, P) \log^3 n)$ total time.
Our approach applies to simple $k$-gons with a linear factor in $k$ in space and time, to pairwise disjoint disks with radii in $[1, k]$, and to unit disks of ply $k$ (see Table~\ref{tab:my_label}). 
Compared to previous approaches for convex hulls that achieve optimal retrieval counts~\cite{bruce2005efficient}, we exponentially improve the runtime per retrieval from polynomial to polylogarithmic. Compared to near-linear time algorithms for convex hulls (or, Delaunay triangulations) \cite{devillers2011delaunay,held2008triangulating,loffler2010delaunay,loffler2025preprocessing,van2010preprocessing,evans2011possible, ezra2013convex}, we significantly reduce the number of retrievals used and broaden the generality of input families.
The latter comes at a cost of a polylogarithmic slowdown if $r(F, P)$ is linear in $n$.

\begin{table}[tb]
    \centering
    \begin{minipage}{\textwidth}
    \begin{tabular}{@{}llllllll@{}}
    \toprule
       Shapes 
       & Overlap & Structures & Preprocess  & Retrievals \hspace{3mm} & Reconstruction time & Source \\
    \midrule
             Unit disks  & No & many & $O(n \log n)$ & $O(n)$ & $O(n)$  & \hspace{-1.5cm}\cite{devillers2011delaunay, evans2011possible ,held2008triangulating, loffler2013unions, loffler2010delaunay, van2010preprocessing} \\
 Disks & $k$ & Delaunay & $O(n \log n)$ & $O(n)$ & \color{black}$O(n \log k)$ rand. & \cite{buchin2011delaunay} \\
 Lines & Yes & hull & $O(n \log n)$ & $O(n)$ & $(n \cdot \alpha(n))$ rand. & \cite{ezra2013convex} \\
            
      \noalign{\medskip}
    
           Intervals & Yes &  sorting & $O(n \log n)$ &  $\Theta( \rfp)$ & $O(\rfp \cdot \log n)$ & \cite{van2019preprocessing} \\
       Axis rect. & No & front & $O(n \log n)$ &  $\Theta( \rfp)$ & $O(\rfp \cdot \log n)$ & \cite{van2022preprocessing} \\      
      Smooth & Yes & front, hull & $O(\poly n) $ & $\Theta(\rfp)$ & $O(\rfp \cdot \poly n )$  & \cite{bruce2005efficient} \\

             \noalign{\medskip}
                    $O(1)$-gons &Yes & hull & $O(n \log^3 n)$ & $\Theta( \rfp)$ & $O(  \rfp \cdot \log^3 n)$ & Thm~\ref{lem:polygons_slow} \\
                    $k$-gons & Yes & hull & $O(kn \log^3 n)$ & $\Theta( \rfp)$ & $O(  \rfp \cdot k \log^3 n)$ & Thm~\ref{thm:polygons_improved} \\
       $[1, k]$-disks & No & hull & $O( kn \log^3 n )$ &  $\Theta( \rfp)$  & $O( \rfp \cdot k \log^3 n)$ & Thm.~\ref{thm:disks} \\
       Unit disks & $k$ & hull & $O(k n \log^3 n)$ & $\Theta(r(F, P))$ & $O( \rfp \cdot k \log^3 n)$ &  Thm.~\ref{thm:disks}  \\
       Unit disks & $k$ & hull & $O(k^3 n)$ & $O(w(F))$ & $O( k^3 \cdot w(F))$ & \cite{loffler2025preprocessing} 
    \end{tabular}
    \caption{
    Previous results that compute a \emph{sorting}, Pareto \emph{front}, convex \emph{hull}, or \emph{Delaunay} triangulation.  
    The bottommost result is simultaneous and independent work from ours.    
}
    \label{tab:my_label}
    \end{minipage}
\end{table}

\subparagraph{Organisation.}
In Section~\ref{sec:preliminaries}, we provide some preliminaries.
Then, in Section~\ref{sec:witness_sets_optimal_query_strategy}, we present our general algorithm for reconstructing the convex hull 
and prove its optimality (\cref{theo:finished}). 
In Section~\ref{sec:kgons}, we present a data structure for simple $k$-gons and give an instance-optimal algorithm that uses $O(kn)$ space and  $O( r(F, P) \cdot k^2 \log^3 (kn) )$ time (Theorem~\ref{lem:polygons_slow}).
In the full version we improve the dependency on $k$ to near-linear (Theorem~\ref{thm:polygons_improved}).
In the full version, we study pairwise disjoint disks with a radius in $[1,k]$, or, unit disks with ply $k$. 
 
\section{Preliminaries}\label{sec:preliminaries}
Recall that an imprecise point set is defined as a family of geometric regions $F = ( R_1, \ldots, R_n )$, and a realisation $P \sim F$ is defined as a sequence $P$ of $n$ points with $p_i \in R_i$.
We explicitly allow for duplicate points in $P$ and do not assume general position.
The algorithm has two phases~\cite{held2008triangulating}: a \emph{preprocessing phase}, where only $F$ is available, and a \emph{reconstruction phase}, in which we may perform a \emph{retrieval} on any $R_i \in F$,  replacing $R_i$ with the point~$p_i$. 
To maintain generality, we require only that each $R_i$ is a point or a closed, connected, bounded region whose boundary is a simple closed piecewise-$C^1$ curve. Non-point regions must have connected interiors and coincide with the closure of their interiors; see Figure~\ref{fig:geometric_regions}.
Next, we define retrievals, convex hulls, vertices, edges, reconstruction strategies, and instance-optimality.

\begin{figure}[b]
    \centering
    \includegraphics{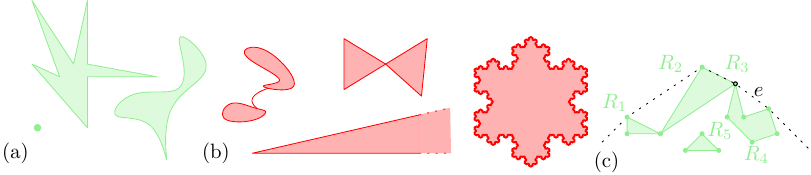}
    \caption{
        (a) Regions may have bends and sharp corners, and, be points.
        (b) Regions may not be unbounded, non-simple, connected by a line, or have infinitely many sharp corners.
        (c) Vertices between regions may coincide, and point regions ($R_3$) may coincide with vertices of other regions. 
    }
    \label{fig:geometric_regions}
\end{figure}

\subparagraph{Retrievals.}
Given a family of regions~$F$, a \emph{retrieval} operation on a non-point region~$R_i$ replaces $R_i$ with its associated point~$p_i$ (thereby updating $F$). 
We write $F \retrieve A$ for the family resulting from retrieving all regions in $A \subseteq F$. 
We denote by $F - A$ the family obtained by removing all regions in $A$ from $F$. 
This removal is used in our analysis when we identify regions whose corresponding points cannot appear on the convex hull.

\subparagraph{Critical assumptions.}
Since all regions in $F$ are closed, we claim that we must assume $P$ may contain duplicates and require that the convex hull of a point sequence includes all collinear and duplicate points. 
To illustrate, let $F$ consist of $n$ identical unit squares, and suppose $P \sim F$ includes four points forming the corners of one square. A lucky algorithm might retrieve these four points first and construct the convex hull~\cite{Kahan1991Model}. 
If we assumed general position or excluded duplicates from the convex hull, this algorithm could then conclude that all remaining points cannot lie on the convex hull and terminate early.
However, such behaviour cannot be guaranteed against an adversary. Since all regions are indistinguishable, an adversary may simply give these four points last. 
Alternatively, if $F$ consisted of open regions, we could assume general position and exclude duplicates.

\subparagraph{Upper quarter convex hull.}
We focus on the \emph{upper quarter convex hull} of $P$. The other three quarters can be constructed analogously, and the full convex hull results from combining them.
Let $\CH(P)$ denote the upper quarter convex hull: the boundary of the smallest convex set containing $P \cup \{(-\infty, -\infty), (+\infty, -\infty)\}$. 
We assume that $(R_{n+1}, R_{n+2}) = (\{(-\infty, -\infty)\}, \{(+\infty, -\infty)\})$, with corresponding points $(p_{n+1}, p_{n+2})$. These are always included in $F$ and $P$.
For any family of regions $F$, we denote by $\OCH(F)$ the \emph{upper quarter outer convex hull}: the boundary of the smallest convex area enclosing all regions in $F$. 
Thus, $\CH(P)$ and $\OCH(F)$ represent convex hulls over point sequences and region families, respectively.
As convex hulls define boundary curves, we may say that a point $p \in \mathbb{R}^2$ lies \emph{on}, \emph{inside}, or \emph{outside} $\CH(P)$ or $\OCH(F)$.

\subparagraph{Vertices and edges.}
We define convex hull vertices and edges of $\OCH(F)$. 
Although intuitive in simpler settings, these concepts require care due to the generality of $F$.

\begin{definition}
    \label{definition:vertex}   
    A point $p \in R$ is a \emph{vertex} of a region~$R$ if $p$ lies on the boundary of $R$ and every open line segment through~$p$ contains a point not on the boundary of $R$.
    For any $R \in F$, $V(R)$ is the (infinite) set of vertices of $R$. 
\end{definition}

Intuitively, the edges of $\OCH(F)$ connect successive vertices lying on its boundary. 
Due to potential overlaps between regions, vertices may coincide. 
As illustrated in Figure~\ref{fig:geometric_regions}(c), a single edge may therefore be considered to connect different pairs of overlapping vertices.
To avoid such degeneracy, we define $V(F) = \bigcup_{R \in F} V(R)$ as a set, and define edges robustly:

 \begin{definition}
    \label{definition:edge}
    We define an \emph{edge} $(s, t)$ of $\OCH(F)$ to be a pair of distinct vertices in $V(F)$ where the subcurve of~$\OCH(F)$ from~$s$ to~$t$ contains no vertices in $V(F)$ other than~$s$ and~$t$.
\end{definition}

\noindent
Note that a disk has infinitely many vertices but no edges under these definitions.
We also define when a region appears on the outer convex hull:

\begin{definition}
    We say a region $R$ \emph{appears on} $\OCH(F)$ if there is a point $p$ on $\OCH(F)$ with $p \in R$.
    (Observe that we may always pick $p$ to be a vertex of $R$.)
\end{definition}

\subparagraph{Reconstruction algorithms.} We will retrieve regions until all realisations $P' \sim F$ yield the same ordering of points along their convex hull. 
We formalise this via a partial order:

\begin{definition}
    Given $P \sim F$, let $\precorder(\CH(P))$ be the partial order on $[n]$ induced by the left-to-right traversal of $\CH(P$), i.e.
    for $p_a, p_b \in P$, we have $a \prec b$ if and only if $p_a$ and $p_b$ both lie on $\CH(P)$ and $p_a$ lies strictly to the left of $p_b$.
\end{definition}

We say a family $F$ is \emph{finished} if all realisations $P_1, P_2 \sim F$ satisfy $\precorder(\CH(P_1)) = \precorder(\CH(P_2))$. 
A \emph{reconstruction algorithm} retrieves some $B \subseteq F$ such that $F \retrieve B$ is finished. 
Let $r(F, P)$ denote the minimum number of retrievals needed by any such algorithm.

\begin{observation}
\label{obs:trivial_lower_bound}
For any fixed $P \sim F$, let $A \subset F$ be a smallest subset of $F$ such that $F \retrieve A$ is finished. Then $|A|$ is a tight lower bound for $r(F, P)$. 
\end{observation}

A reconstruction algorithm is \emph{instance-optimal} if for all inputs $(F, P)$ it retrieves $\Theta(r(F, P))$ regions before $F$ is finished.
We distinguish two types of reconstruction algorithms:
\begin{itemize}
    \item A \emph{reconstruction strategy} is any such algorithm, analysed only by the number of retrievals.
    \item A \emph{reconstruction program} is a reconstruction algorithm executed on a pointer machine or RAM, and is analysed by both retrievals and instructions. 
\end{itemize}

\subparagraph{The previous instance-optimal reconstruction strategy.}
Bruce, Hoffmann, Krizanc, and Raman~\cite{bruce2005efficient} present an instance-optimal reconstruction strategy for region families $F$ consisting of points and closed piecewise-$C^1$ regions. 
They do not present a corresponding reconstruction program. 
When $F$ consists of disks or polygons, a reconstruction program may be derived, albeit with unspecified polynomial-time costs per retrieval. 
We further discuss their algorithm in the full version.
We rely on a key concept of their paper, i.e. a witness.

\begin{definition}
    Let $F$ be a family of regions and let $A \subseteq F$. 
    Any $A \subset F$ is a \emph{witness} if for all $P' \sim ( F - A)$, the family $A \cup P'$ is not finished.     
\end{definition}

\noindent
We observe that a reconstruction strategy is instance-optimal if, at each step, it 
retrieves all regions from a constant-size witness set of~$F$.

\subsection{Explicit dynamic planar convex hull}

We rely on a data structure to maintain the upper quarter convex hull of a dynamic point set.
Overmars and van Leeuwen~\cite{DBLP:conf/stoc/OvermarsV80} provide the best-known solution, achieving deterministic update time $O(\log^2 n)$. 
Their data structure, known as the \emph{Partial Hull Tree} (PHT), is now textbook material. We follow their terminology:

\begin{definition}[PHT from~\cite{DBLP:conf/stoc/OvermarsV80}]
\label{definition:pht}
Given a two-dimensional point set $S$, the \emph{Partial Hull Tree} stores $S$ in a leaf-based balanced binary tree $T$ (sorted by $x$-coordinates).
For each interior node $\nu \in T$, denote by $\CH(\nu)$ the (upper quarter) convex hull of all points in the leaves of the subtree rooted at $v$.
For each $\nu \in T$, with children $(x, y)$ and parent $w$, the PHT stores:
\begin{itemize}
    \item  The \emph{bridge} $e(\nu)$  (the unique edge in $\CH(\nu)$ that is not in $\CH(x)$ and $\CH(y)$), and
    \item  a \emph{concatenable queue} $\mathbb{E}^*(\nu)$ (a balanced tree of the edges in $CH(\nu) - CH(w)$). 
\end{itemize}
  We denote by $\mathbb E(\nu)$  a balanced tree over $CH(\nu)$ and by $V[\nu]$ the vertices in the subtree of $\nu$.
\end{definition}

\noindent
Note that for the root $\rho$ of $T$, $\mathbb E(\rho)= \mathbb{E}^*(\rho)$. 

\section{Witnesses and an instance-optimal reconstruction strategy}\label{sec:witness_sets_optimal_query_strategy}

We introduce a new geometric classification over the edges of $\OCH(F)$. 
This leads to an alternative instance-optimal reconstruction strategy, detailed in Algorithm~\ref{alg:new_strategy}. 
Our results hold in full generality for any family $F$ of points and closed regions whose boundaries are simple, closed, piecewise-$C^1$ curves.
In particular, each region is connected and bounded, and any non-point region has a connected interior equal to its closure; see \cref{fig:geometric_regions}.
We define a classification over vertex pairs $(s, t)$, which we apply to edges of $\OCH(F)$.
However, we formulate the classification for arbitrary vertex pairs, as this generality will be useful for our data structures.
We say that two regions $R_a, R_b \in F$ are \emph{strictly vertically separated} if there exists a vertical line $\ell$ where $R_a$ and $R_b$ lie in opposite open half-planes defined by $\ell$.

\begin{figure}[H]
    \centering
    \includegraphics[page=2]{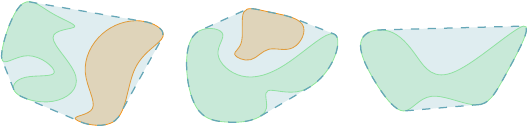}
    \caption{
    Three examples of a band (blue) between regions (green, brown).
    }
    \label{fig:definition_band}
\end{figure}
    
\begin{definition}[Band. Fig.~\ref{fig:definition_band}] \label{def:band}
    For any pair of regions $(R_a, R_b)$ (where $a = b$ is allowed)
we define $\band(R_a, R_b)$ as the area enclosed by the (full) outer convex hull of $\{ R_a, R_b \}$.
\end{definition}

\begin{definition}[Fig.~\ref{fig:canonical-dividing-occupied}]
\label{def:edge_label}
    We say $(s,t)\in V(R_a)\times V(R_b)$ with $s\neq t$ and $R_a,R_b \in F$, is:
    \begin{itemize}
        \item \emph{canonical} in $F$ if the following holds for all $x \in \{ s, t \}$.  Either:
        \begin{itemize}
            \item $x$ is exclusively a vertex of point regions, or 
            \item $x$ is a vertex of a unique region that is not a point region. 
        \end{itemize}
        \item \emph{dividing} in $F$ if $(s, t)$ is canonical and either:
        \begin{itemize}
            \item $a =  b$, or
            \item $R_a,R_b$ are strictly vertically separated. 
        \end{itemize}
        \item \emph{occupied} in $F$ if $(s, t)$  is dividing and:
        \begin{itemize}
            \item $\band(R_a, R_b)$ contains a vertex in $V(F - R_a - R_b) - \{ s  \} -  \{ t \}$.
        \end{itemize}
    \end{itemize}
\end{definition}

\noindent
In our reconstruction strategy, we retrieve regions until no non-canonical, non-dividing, or occupied edges remain. This alone does not guarantee that $F$ is finished, so we introduce a final characterisation based on convex chains. 

\begin{definition}[Spanning chain. Fig.~\ref{fig:spanning-chain}]
    \label{def:spanning}
A convex chain $C = (q, r, \ldots, s, t)$ is \emph{spanning} in $F$ if the following three conditions hold:
    \begin{itemize}
        \item all edges on $C$ are dividing in $F$,
        \item $q \in V(R_a)$, $r, s \in V(R_b)$, $t \in V(R_c)$ where $R_b$ intersects the inside of $\OCH(\{R_a, R_c\})$, 
        \item all vertices of $C$ between and including $r$ and $s$ are in $V(R_b)$.
    \end{itemize}
\end{definition}

\begin{figure}[H]
    \centering
    \includegraphics{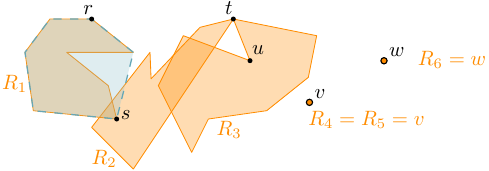}
    \caption{ We illustrate all (sub) cases of vertex pairs corresponding to Definition~\ref{def:edge_label}:
    \begin{itemize}
        \item All pairs that include $t$ are non-canonical, as $t$ is a vertex of more than one non-point region.
        \item All pairs in $\{r,s,u,v, w\}^2$ are all canonical, as $v$ is exclusively a vertex of two point regions.
        \item The pair $(u,v)$ is non-dividing, as $R_3$ and $R_4$ are not vertically separated.
        \item  The pairs $(s, u)$, $(r,u)$, $(r,v)$, $(s,v)$, and any pair in $\{ r, s, u, v \} \times \{ w \}$ are dividing because their regions are vertically separated. $(r,s)$ is dividing because they are vertices of the same region.
        \item The pair     
        $(r,s)$ is not occupied since $\band(R_1,R_1)$ only contains vertices of $R_1$. The pair $(v,w)$ is not occupied since $\band(R_4,R_6)$ only contains vertices $v$ and $w$.
        \item  The pair $(u, w)$ is occupied as $V(F - R_3 - R_6)$ contains the vertex $t \in R_2$. The pairs $(s, u)$, $(r, u)$, $(r, v)$, $(s, v)$, $(r, w)$, and $(s, w)$ are also occupied.
    \end{itemize}
        }
    \label{fig:canonical-dividing-occupied}
\end{figure}

\begin{figure}[t]
    \centering
    \includegraphics[page=2]{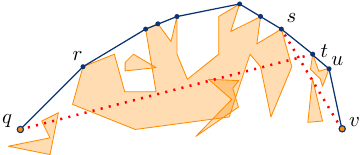}
    \caption{The convex chain from $q$ to $t$ is spanning. The chain from $r$ to $v$ is not spanning.
    }
    \label{fig:spanning-chain}
\end{figure}

\noindent
We use $\OCH(\{R_a, R_c\})$ rather than $\band(R_a, R_c)$ in Definition~\ref{def:spanning} to ensure meaningful behaviour even when $R_a$ and $R_c$ are point regions (since $\band(R_a, R_c)$ would then be empty).
Algorithm~\ref{alg:new_strategy} processes edges of $\OCH(F)$ by a priority based on Definitions~\ref{def:edge_label} and \ref{def:spanning}:

\[
\textnormal{non-canonical } \, \gg  \, \textnormal{ canonical but non-dividing } \,  \gg\,  \textnormal{ occupied } \,  \gg \, \textnormal{ in spanning chain.} 
\]

\begin{algorithm}
    \caption{A data structure friendly reconstruction strategy: Reconstruct$(F)$.}
    \label{alg:new_strategy}
    \begin{algorithmic}
        \IF{ $\exists$ an edge $(s, t)$ of $\OCH(F)$ that is non-canonical in $F$}
            \STATE Let $s \in V(R_a)$ and $s \in V(R_b)$ (or $t$) where $R_a$ is not a point region
            \STATE Reconstruct$(F \retrieve \{R_a, R_b \})$ 
        \ELSIF{ $\exists$ an edge $(s, t)$ of $\OCH(F)$ that is canonical but not dividing in $F$}
            \item Let  $s \in V(R_a)$ and $t \in V(R_b)$ for $a \ne b$ 
            \STATE Reconstruct$(F \retrieve \{R_a, R_b \})$
        \ELSIF{ $\exists$ an edge $(s, t)$ of $\OCH(F)$ that is occupied in $F$}
             \item Let  $s \in V(R_a)$ and $t \in V(R_b)$ for $a \ne b$
             \item Let $R_i$ with $R_i \ne R_a$ and $R_i \ne R_b$ be a region with $V(R_i) \cap \band(R_a, R_b) \neq \emptyset$
            \STATE Reconstruct$(F \retrieve \{R_a, R_b, R_i \})$
        \ELSIF{ $\exists$ a contiguous subchain $C = (q, r, \ldots,  s, t)$ of $\OCH(F)$ spanning in $F$}
        \STATE Let $R_a, R_b, R_c$ be as in the definition of spanning chains
            \STATE Reconstruct$(F \retrieve \{R_a, R_b, R_c\})$
        \ELSE
            \STATE \textbf{Return} $\OCH(F)$
        \ENDIF
    \end{algorithmic}
\end{algorithm}

\noindent
In the following two subsections, we first show that each case in \cref{alg:new_strategy} retrieves a witness. Then, we prove that  if \cref{alg:new_strategy} is instance-optimal.  

\subsection{ \texorpdfstring{Proving that \cref{alg:new_strategy} retrieves only witnesses}{Proving that the strategy retrieves only witnesses}}
\label{section:witnesses}

The first case of \cref{alg:new_strategy} considers a non-canonical edge, which implies that there is a point on $\OCH(F)$ that lies in two regions that are not both point regions.

\begin{restatable}{lemma}{doublepointWitness}

\label{lemm:wit_doublepoint}
    Let $s \in \OCH(F)$ be a point that lies in two regions $R_a, R_b$ that are not both point regions.
    Then $\{R_a, R_b\}$ is a witness.
\end{restatable}
\begin{proof}
    Let $P' \sim (F - R_a - R_b)$ be arbitrary.
    Since $R_a$ and $R_b$ appear on $\OCH(F)$, there are vertices $q, r$ on $\OCH(F)$ with $q \in R_a$ and $r \in R_b$.
    If $q = r$, \cref{lemm:wit_doublepoint} shows that $\{R_a, R_b\}$ is a witness.
    Otherwise, set $p_a = q, p_b = r$, then $R_a$ and $R_b$ appear at distinct points on $\CH(P' \cup \{p_a, p_b\})$.
    On the other hand, as $R_a$ and $R_b$ are not vertically separated, there exist $s \in R_a, t \in R_b$ with the same $x$-coordinate. Set $p_a' = s, p_b' = t$, then $R_a$ and $R_b$ cannot appear at distinct points
    on $\CH(P' \cup \{p_a', p_b'\})$. Indeed, if $s = t$, then $R_a$ and $R_b$ form duplicate points,
    and if $s \ne t$, then at most one of $R_a$ and $R_b$ appears on $\CH(P' \cup \{p_a', p_b'\})$.
    This shows that $P' \cup \{R_a, R_b\}$ is not finished.
\end{proof}

The second case of \cref{alg:new_strategy} considers an edge $(s,t)$ of $\OCH(F)$ that is canonical but not dividing in $F$. Let $s \in V(R_a)$ and $t \in V(R_b)$. Because $(s,t)$ is canonical but not dividing it holds that $a \neq b$ and $R_a$ and $R_b$ are distinct regions that are not strictly vertically separated. The following lemma implies that $\{R_a,R_b\}$ is indeed a witness.

\begin{restatable}{lemma}{verticallyseparatedWitness}
\label{lem:wit_vertically_separated}
    Let $R_a, R_b \in F$ with $R_a \ne R_b$ appear on $\OCH(F)$.
    If $R_a$ and $R_b$ are not strictly vertically separated, then $\{R_a, R_b\}$ is a witness.
\end{restatable}

\begin{proof}
    Let $P' \sim (F - R_a - R_b)$ be arbitrary.
    Since $R_a$ and $R_b$ appear on $\OCH(F)$, there are vertices $q, r$ on $\OCH(F)$ with $q \in R_a$ and $r \in R_b$.
    If $q = r$, \cref{lemm:wit_doublepoint} shows that $\{R_a, R_b\}$ is a witness.
    Otherwise, set $p_a = q, p_b = r$, then $R_a$ and $R_b$ appear at distinct points on $\CH(P' \cup \{p_a, p_b\})$.
    On the other hand, as $R_a$ and $R_b$ are not vertically separated, there exist $s \in R_a, t \in R_b$ with the same $x$-coordinate. Set $p_a' = s, p_b' = t$, then $R_a$ and $R_b$ cannot appear at distinct points
    on $\CH(P' \cup \{p_a', p_b'\})$. Indeed, if $s = t$, then $R_a$ and $R_b$ form duplicate points,
    and if $s \ne t$, then at most one of $R_a$ and $R_b$ appears on $\CH(P' \cup \{p_a', p_b'\})$.
    This shows that $P' \cup \{R_a, R_b\}$ is not finished.
\end{proof}

The third case of \cref{alg:new_strategy} considers an occupied edge. The next lemma shows that any choice of region that has a vertex in the occupied band gives rise to a witness.
\begin{restatable}{lemma}{occupiedWitness}\label{lem:wit_occupied}
    Let $(s, t)$ be an edge of $\OCH(F)$ that is occupied in $F$.
    Let $s \in V(R_a)$ and $t \in V(R_b)$ for $a \ne b$.
    Let $R_i$ be a region not equal to $R_a$ or $R_b$
    with a vertex $q \in V(R_i) \cap \band(R_a, R_b)$. 
    Then $\{R_a, R_b, R_i\}$ is a witness.
\end{restatable}

\begin{proof}
Figure~\ref{fig:occupied-edge} illustrates the proof.
Let $P' \sim (F - R_a - R_b - R_i)$. 
    Since $(s,t)$ is occupied (and canonical), both $s$ and $t$ are not in~$R_i$, so $q \notin \{s, t\}$.
    First, set $p_a = s, p_b = t, p_i = q$, then $R_a, R_b$ appear consecutively on $\CH(P' \cup \{p_a, p_b, p_i\})$. 
    Indeed, as $(s, t)$ is an edge of $\OCH(F)$, also $(p_a, p_b)$ is an edge of $\CH(P' \cup \{p_a, p_b, p_i\})$.
    Moreover, $p_i \ne p_a$ and $p_i \ne p_b$.

    \begin{figure}[ht]
        \centering
        \includegraphics{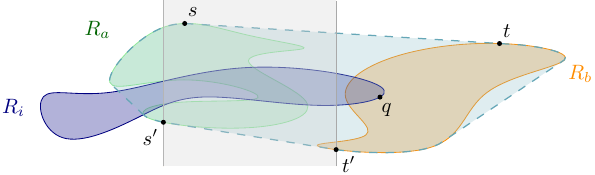}
        \caption{The region $R_i$ has a vertex $q$ inside $\band(R_a,R_b)$. The points $s'$ and $t'$ define the lower tangent. The proof considers two cases based on the location of $q$ with respect to the vertical slab.}
        \label{fig:occupied-edge}
    \end{figure}

    On the other hand, let $s' \in R_a, t' \in R_b$ be the outermost supporting points of the lower tangent of $R_a$ and $R_b$. Then $q$ lies on or above the line $s't'$.
    Set $p_a' = s', p_b' = t', p_i' = q$.
    If~$q$ lies outside of the vertical slab defined by $s'$ and $t'$, then it is strictly above the line $s' t'$ and
    $R_a$ and $R_b$ do not both appear on $\OCH(P' \cup \{p_a', p_b', p_i'\})$. 
    Otherwise, $q$ lies on or above the segment $s' t'$.
    But then either $(p_a', p_b')$ is not an edge of $\CH(P' \cup \{p_a', p_b', p_i'\})$,
    or $p_i' = p_a'$ or $p_i' = p_b'$.
    In all cases, $\CH(P' \cup \{p_a, p_b, p_i\}) \ne \CH(P' \cup \{p_a', p_b', p_i'\})$.
\end{proof}

\noindent
The final case of \cref{alg:new_strategy} considers any remaining spanning subchains of $\OCH(F)$.

\begin{restatable}{lemma}{subchainWitness}
\label{lem:wit_subchain}
    Let $(q, r, \dots, s, t)$ be a contiguous subchain of $\OCH(F)$ that is spanning in $F$. 
    Let $R_a, R_b, R_c$ be as in Definition~\ref{def:spanning}.
    Then $\{R_a, R_b, R_c\}$ is a witness.
\end{restatable}

\begin{proof}
    Let $P' \sim (F - R_a - R_b - R_c)$.
    Set $p_a =q, p_b = s, p_c = t$, then $R_b$ appears on $\CH(P' \cup \{p_a, p_b, p_c\})$.
    On the other hand,
    let $u \in R_a, w \in R_c$ be the outermost supporting points of the upper tangent of $R_a, R_c$.
    Let $v \in R_b$ be in the inside of $\OCH(R_a, R_c)$.
    Then~$v$ lies strictly below the line $u w$.
    More precisely, since $(q, r, \dots, s, t)$ is spanning,
    $R_b$ lies strictly to the right of $R_a$, and strictly to the left of $R_c$,
    hence $v$ lies below the line \emph{segment}~$u w$.
    Set $p_a' = u, p_b' = v, p_c' = w$, then $R_b$ does not appear on $\CH(P' \cup \{p_a', p_b', p_c'\})$.
\end{proof}

\subsection{\texorpdfstring{Proving that \cref{alg:new_strategy} is instance-optimal}{Proving that the algorithm is instance-optimal}}

We now prove that \cref{alg:new_strategy} is instance-optimal. That is, the algorithm retrieves the minimal number of regions (up to constant factors) necessary to determine the ordering of points on the convex hull, for any realisation $P \sim F$.
We first observe that \cref{alg:new_strategy} only terminates when $F$ is \emph{terminal}, meaning: every edge of $\OCH(F)$ is dividing,
no edge of $\OCH(F)$ is occupied, and
no contiguous subchain of $\OCH(F)$ is spanning.

For the remainder of this section, let $A$ denote the multiset of regions that appear on $\OCH(F)$. If $F$ is terminal, the regions in $A$ satisfy the following structure:

\begin{lemma} \label{lemm:order}
    If $F$ is terminal, then all regions in $A$ are strictly vertically separated from each other, with the exception of pairs of regions that are duplicates of the same point. 
\end{lemma}
\begin{proof}
    Walk along $\OCH(F)$ from left to right, considering edges $(s, t)$ with $s \in R_a$, $t \in R_b$, $a \ne b$. As $F$ is terminal, this edge is dividing, hence $R_a$ lies strictly to the left of $R_b$. All previously encountered regions lie strictly left of $R_a$, hence also strictly to the left of $R_b$. Moreover, this edge is not occupied, so $R_b$ is disjoint from all regions not equal to $R_b$.
\end{proof}

With Lemma~\ref{lemm:order}, we define for each region $R$ that appears on $\OCH(F)$ its left (or right) neighbour as the first region on $\OCH(F)$ that lies \emph{strictly} left (or right) of $R$. Note that two duplicate point regions on $\OCH(F)$ have the same neighbours. Next, we show that the points corresponding to regions in $A$ always appear on the convex hull (Lemma~\ref{lemm:outer_is_hull}), and that no points corresponding to regions in $F - A$ appear on the convex hull (Lemma~\ref{lemm:inner_regions_not_on_CH}).

\begin{lemma}\label{lemm:outer_is_hull}
   If $F$ is terminal, then for any $P'\sim A$ each point in $P'$ lies on $\CH(P')$.
\end{lemma}
\begin{proof}
    Let $P' \sim A$ and $R_b \in A - \{(-\infty, -\infty), (\infty, -\infty)\}$ be arbitrary. Let $R_a$ and $R_c$ be the left and right neighbour of $R_b$ on $\OCH(F)$.
    Observe that these neighbours are well-defined as $(-\infty,-\infty)$ and $(\infty,-\infty)$ lie on $\OCH(F)$.
    \begin{figure}
        \centering
        \includegraphics[page=4]{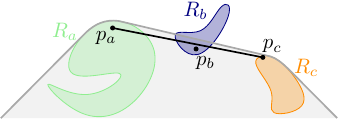}
        \caption{$p_b$ must be above $p_a p_c$, as otherwise there is a spanning contiguous subchain of $\OCH(F)$.}
        \label{fig:outer_is_hull}
    \end{figure}
    Suppose by contradiction that $p_b$ lies below the segment $p_{a} p_{c}$.
    Then $R_b$ intersects the inside of $\OCH(\{R_{a}, R_{c}\})$, see Figure~\ref{fig:outer_is_hull},
    so there is a contiguous subchain of $\OCH(F)$ that is spanning in $F$.
    This contradicts $F$ being terminal.
    Since $R_b$ was arbitrary, this shows that any three consecutive points in $P'$ form a convex-down angle.
    Hence, $P'$ spans a convex polygon, so all points in $P'$ appear on $\CH(P')$.
\end{proof}

\begin{lemma}\label{lemm:inner_regions_not_on_CH}
    If $F$ is terminal, then for any $P \sim F$ only points of regions in $A$ lie on~$\CH(P)$.
\end{lemma}

\begin{proof}
    Let $P' \sim A$ be arbitrary. 
    Suppose by contradiction that there is a region $R_i \in F - A$ with a point $p_i \in R_i$ on or outside of $\CH(P')$.
    Since the inside of $\CH(P')$ is convex, we may require $p_i$ to be an extreme point of $R_i$, so $p_i \in V(R_i)$.
    By Lemma~\ref{lemm:order} and Lemma~\ref{lemm:outer_is_hull}, there exist consecutive regions $R_a, R_b \in A$ for which $p_i$ lies on or above $p_a p_b$.
    As $p_a p_b$ lies fully inside $\band(R_a, R_b)$, the point $p_i$ lies in or directly above $\band(R_a, R_b)$.
    We make a case distinction, where $p_i$ lies directly above $R_b$ (or $R_a$), or, on the upper tangent of $R_a, R_b$.
    
    Suppose first that $p_i$ lies directly above a point $q \in R_b$. 
    Let $R_a, R_c$ be the left and right neighbour of $R_b$ on $\OCH(F)$, respectively.
    Let $r$ be the point on $\OCH(F)$ directly above $p_i$.
    Then $r \in \band(R_b, R_c)$ (or $r \in \band(R_a, R_b)$),
    hence also $p_i \in \band(R_b, R_c)$ as bands are convex and $p_i$ lies on $q r$.
    Hence, $F$ is not terminal. Similarly, $p_i$ being directly above a point in $R_a$ implies that $F$ is not terminal.
    Furthermore, $p_i$ cannot lie above the upper tangent of $R_a, R_b$
    since $R_i \notin A$.
    This shows that $p_i$ does not lie above $\band(R_a, R_b)$, hence
    $p_i \in \band(R_a, R_b)$, so $F$ is not terminal.
\end{proof}

\begin{theorem} \label{theo:finished}
    \cref{alg:new_strategy} is an instance-optimal reconstruction strategy.
\end{theorem}
\begin{proof}
First, we show that when  \cref{alg:new_strategy} terminates, the set $F$ is finished.
Let $F$ be terminal. Let $A \subseteq F$ be the multiset of regions that appear on $\OCH(F)$.
Let $P' \sim A$ and $Q' \sim F - A$ be arbitrary.
By \cref{lemm:inner_regions_not_on_CH}, no points in $Q'$ are on $\CH(P' \cup Q')$.
By Lemma~\ref{lemm:outer_is_hull}, the points in $P'$ all appear on $\CH(P' \cup Q')$.
Moreover, together with Lemma~\ref{lemm:order}, this implies that for all $P_1 \sim A$ and $P_2 \sim A$, $\precorder(\CH(P_1)) = \precorder(\CH(P_2))$. So, $F$ is finished. Hence, \cref{alg:new_strategy} is a correct reconstruction strategy.

We now show instance-optimality.
By Lemmas~\ref{lemm:wit_doublepoint}, \ref{lem:wit_vertically_separated}, \ref{lem:wit_occupied} and \ref{lem:wit_subchain}, the algorithm retrieves a constant-size witness in each iteration.
Let $X_1, \dots, X_k$ be the sets of non-point regions of these witnesses. Since the witnesses have constant size, \cref{alg:new_strategy} does $O(k)$ retrievals.
A retrieval turns a region into a point region, hence the $X_i$ are pairwise disjoint.
On the other hand, if $(F\retrieve A)$ is finished for some $A \subseteq F$,
then $A$ needs to contain at least one element from each $X_i$ (by the definition of witness).
As the $X_i$ are disjoint, this forces $|A| \ge k$.
Hence, \cref{alg:new_strategy} is instance optimal by Observation~\ref{obs:trivial_lower_bound}.
\end{proof}

\section{ \texorpdfstring{Regions that are simple $k$-gons}{Regions that are simple k-gons}}
\label{sec:kgons}

We present a reconstruction program that executes our strategy from Algorithm~\ref{alg:new_strategy} in polylogarithmic time per retrieval. To this end, we restrict $F$ to a family of $n$ (possibly overlapping) simple polygons, each with at most $k$ vertices. Recall that $V(F) = (v_1, \ldots, v_m)$ denotes the set of all $m \in O(kn)$ \emph{distinct} vertices in $F$. 
After each retrieval, we update~$F$ by replacing $R_i$ with $p_i$, and update $V$ by deleting $O(k)$ vertices from~$V$ and adding $p_i$ to $V$.

\begin{observation}
\label{obs:points_to_polygon}
    For any family of regions $F$, $\OCH(F)$ has the same edges as $\CH(V(F))$.
\end{observation}

\noindent
Consider a PHT $T$ of $V(F)$ (Definition~\ref{definition:pht}). After one retrieval, let $T'$ be the updated PHT. We define the \emph{recourse set} as the set symmetric difference of all bridges in $T$ and all bridges in $T'$. 
\emph{Recourse} is defined as the maximum size, across all possible retrievals, of any recourse set.
A PHT has $O(\log m)$ recourse per update, and hence $O(k \log m)$ recourse per retrieval. 

\subparagraph{Augmenting the Partial Hull Tree.}
Within each leaf of $T$, corresponding to a vertex $v \in V(F)$, we maintain two doubly linked lists.
The \emph{point region list} stores all point regions that are $v$. The \emph{non-point region list} stores the other regions $R \in F$ where $v \in V(R)$. 
Recall that Algorithm~\ref{alg:new_strategy} retrieves regions corresponding to the endpoint of edges that are:

\[
    \textnormal{non-canonical } \, \gg  \, \textnormal{ canonical but non-dividing } \,  \gg\,  \textnormal{ occupied } \,  \gg \, \textnormal{ in spanning chain}.
\]

\noindent
For every node $\nu \in T$ with bridge $e(\nu) = (s, t)$, we maintain pointers to the leaves containing $s$ and $t$, along with Boolean flags indicating whether $e(\nu)$ is canonical or dividing. 
An update to $F$ may make $O(n)$ edges occupied. So instead, we maintain a different property (Fig.~\ref{fig:spanning-chain-hit}):

\begin{figure}
    \centering
    \includegraphics[page=3]{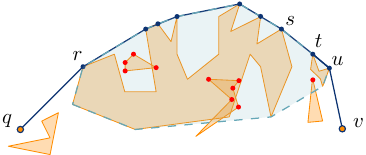}
    \caption{The convex chain from $q$ to $v$ is hit in $F$ since $(s, t)$ is occupied due to the red vertices.
    }
    \label{fig:spanning-chain-hit}
\end{figure}

\begin{definition}
    \label{def:hit}
A convex chain $C = (q, r, \ldots, s, t, \ldots, u, v)$ of vertices is \emph{hit} in $F$ if:
\begin{itemize}
    \item all edges of $C$ are dividing, 
    \item the edge $(s, t)$ is occupied in $F$, 
    \item $s \in V(R_a), t \in V(R_b)$ and $a \ne b$, and
    \item $q \notin V(R_a)$, $r, \dots, s \in V(R_a)$, $t, \dots, u \in V(R_b)$ and $v \notin V(R_b)$.
\end{itemize}
\end{definition}

Recall Definition~\ref{definition:pht}, where for  $\nu \in T$, $\mathbb{E}(\nu)$ denotes a balanced binary tree on $\CH(\nu)$.
The PHT stores in $\nu$ a \emph{concatenable queue}~$\mathbb{E}^*(\nu)$ which is $\mathbb{E}(\nu)$ minus all edges that are in $\mathbb{E}(w)$ where $w$ is the parent of $\nu$.
We further augment the data structure by maintaining four balanced trees associated with~$\mathbb{E}^*(\nu)$, where edges are ordered by their appearance in $\mathbb{E}^*(
\nu)$:

\begin{enumerate}
    \item a tree $\Lambda^*(\nu)$ storing all edges $(s, t) \in \mathbb{E}^*(\nu)$ that are non-canonical in $F$,
    \item a tree storing all edges $(s, t) \in \mathbb{E}^*(\nu)$ that are canonical and non-dividing in $F$, 
      \item a tree storing all contiguous subchains of $\mathbb{E}^*(\nu)$ that are spanning in $F$,
    \item a tree storing all contiguous subchains of $\mathbb{E}^*(\nu)$ that are hit in $F$. 
\end{enumerate}

\noindent
We only give the first tree a name since all others are maintained in similar fashion. 
We denote by $\Lambda_\nu$ a balanced tree on all non-canonical edges in $\mathbb{E}(\nu)$.

\begin{observation}[by Definition~\ref{def:edge_label}]
\label{obs:canonical}
For any canonical bridge $(s, t)$ in $T$, there exists a unique area  $R(s)$ that contains $s$. $R(s)$ is either a point, or a unique region in $F$ and can be found in $O(1)$ time by checking the regions lists stored at one of the leaves that $(s, t)$ points to. 
\end{observation}

We store for every region in $F$ their maximum and minimum $x$-coordinate, this way we can test if a pair of regions is vertically separated in constant time. 

\subparagraph{Candidate chains.} Our data structure maintains for each $\nu \in T$, a collection of special edges and special chains. 
We observe that these chains have a special structure:

\begin{restatable}{definition}{candidate}
    Let $\nu \in T$ and consider $\mathbb{E}(\nu)$ (not $\mathbb{E}^*(\nu)$).   Any contiguous subchain $C = (q, s, \ldots, r)$ of $\mathbb{E}(\nu)$ is a \emph{candidate chain} if all edges of $C$ are dividing in $F$ and:
    \begin{itemize}
        \item $C = (q, s, r)$ with $q \in V(R_a)$, $s \in V(R_b)$, $r \in V(R_c)$,  and $q, r \not \in V(R_b)$, or,
        \item all interior vertices of $C$ are in $V(R_b)$ and $q, r \not \in V(R_b)$,

        (note that this implies that all interior vertices are not in $V(R_x)$ for $x \neq b$).
    \end{itemize}
\end{restatable}

\begin{observation}
    Any subchain of $\mathbb{E}(\nu)$ that is  \emph{spanning} in $F$ is also a candidate chain.
\end{observation}

\begin{observation}
    Any subchain $(q, r, \ldots, s, t, \ldots, u, v)$ of $\mathbb{E}(\nu)$ is \emph{hit} in $\nu$ only if $(q, r, \ldots, s, t)$ and $(s, t, \ldots, u, v)$ are candidate chains. 
\end{observation}

\begin{lemma}
    \label{lemm:candidate_chain_find}
    Let $\nu \in T$.
    There are at most two candidate chains that contain an edge $(s, t)$ of $\mathbb{E}(\nu)$ and, given $\mathbb{E}(\nu)$ and our data structure, we can find these in $O(k)$ time.
\end{lemma}

\begin{proof}
    Suppose two candidate chains $C$ and $C'$ share $(s, t)$. By the definition of candidate chains, $(s, t)$ must be the final edge of $C$ and the first edge of $C'$.
    First, identify the up to two length-3 chains in $\mathbb{E}(\nu)$ containing $(s, t)$. Check each such chain in $O(1)$ time for the dividing property by consulting the Boolean flags. If all edges are dividing, we apply Observation~\ref{obs:canonical} to test whether the chain satisfies the conditions of a candidate chain.

    Next, consider candidate chains of length $\ge 4$ in which $(s, t)$ is an interior edge. Check, in $O(1)$ time via Observation~\ref{obs:canonical}, whether both $s$ and $t$ lie in a unique region $R_b$. If so, perform an in-order traversal of $\mathbb{E}(\nu)$—moving up to $O(k)$ steps left from $s$ and up to $O(k)$ steps right from $t$—to find the maximal contiguous chain with all edges dividing and interior vertices contained in $V(R_b)$. The chain ends when we encounter an edge whose endpoints do not lie entirely in $V(R_b)$. This identifies the unique candidate chain $C$ of length $\ge 4$ containing $(s, t)$ as an interior edge.
    To find chains where $(s, t)$ is the first or last edge, apply the above procedure to the edges immediately before and after $(s, t)$. 
\end{proof}

\subparagraph{Update time.}
After each retrieval, we update bridges for all nodes along $O(k)$ root-to-leaf paths in $T$. For each bridge $(s, t)$, we update its leaf pointers with constant overhead.
Corollaries~\ref{cor:canonical}, \ref{cor:dividing},  \ref{cor:chain}, and~\ref{cor:occupied} show how to determine whether a bridge is canonical, dividing, part of a spanning chain, or part of a hit chain in $F$, respectively, in $O(k \log^2 m)$ time. Since the recourse per retrieval is $O(k \log m)$, this yields an overall update time of $O(k^2 \log^3 m)$.

\begin{observation}[by Definition~\ref{def:edge_label}]
    \label{obs:canonical_test}
    For any node $\nu \in T$, its bridge $e(\nu) = (s, t)$ is \emph{canonical} in $F$ if and only if both leaves containing $s$ and~$t$ meet one of the following conditions:
    \begin{itemize}
        \item The non-point region list is empty, or,
        \item The point region list is empty and the region list contains exactly one region. 
    \end{itemize}
\end{observation}

\begin{corollary}
    \label{cor:canonical}
   We can dynamically maintain, in $O(k \log^2 m)$ time, for each node $\nu \in T$ a balanced binary tree $\Lambda^*(\nu)$ of all edges $(s, t) \in \mathbb{E}^*(\nu)$ that are non-canonical in $F$.     
\end{corollary}

\begin{proof}
Test, using Observation~\ref{obs:canonical_test}, whether any bridge is canonical in constant time. Since the recourse is $O(k \log m)$, the set of all bridges can be updated without incurring asymptotic overhead.

    To maintain $\Lambda^*(\nu)$ for each node $\nu \in T$, we follow the technique from~\cite{DBLP:conf/stoc/OvermarsV80} for concatenable queues.
    For any node $\nu$ with child $x$, compute $\Lambda(x)$ from $(\Lambda(\nu),  e(\nu), \Lambda^*(x))$ in $O(\log m)$ time by splitting $\Lambda(\nu)$ at the bridge $e(\nu)$, and joining the result with $\Lambda^*(x)$. 
    Traverse all $O(k)$ updated root-to-leaf paths, using the split operation, in $O(k \log^2 m)$ total time. 
    Then, traverse the paths bottom-up.
    At each node $\nu$ with children $x$ and $y$, we have access to the newly updated trees $\Lambda(x)$ and $\Lambda(y)$.
    Compute $(\Lambda(\nu), \Lambda^*(x), \Lambda^*(y))$ by splitting $\Lambda(x)$ and $\Lambda(y)$ at the endpoints of the new bridge $e(\nu)$ and joining the result.
\end{proof}

\begin{corollary}
    \label{cor:dividing}
     We can dynamically maintain in $O(k \log^2 m)$ time for each node $\nu \in T$ a balanced binary tree on all edges $(s, t) \in \mathbb{E}^*(\nu)$ that are canonical but non-dividing in $F$.   
\end{corollary}

\begin{proof}
The set of bridges has $O(k \log m)$ recourse. 
Once a bridge is known to be canonical, we apply Observation~\ref{obs:canonical} to determine the unique areas $R(s)$ and $R(t)$ containing $s$ and $t$. We then test in constant time whether $R(s)$ and $R(t)$ are vertically separated, using stored $x$-coordinates. The tree is maintained using the same procedure as in Corollary~\ref{cor:canonical}.
\end{proof}

\begin{lemma}
    \label{lemm:chaincomputation}
     For any $\nu \in T$, given $e(\nu) = (s, t)$ and 
     $\mathbb{E}(\nu)$, we may find the at most two contiguous subchains $C$ of $\mathbb{E}(\nu)$ with $s, t \in C$ that are spanning in $F$ in $O( k \log m)$ time.
\end{lemma}

\begin{proof}
Apply Lemma~\ref{lemm:candidate_chain_find} to obtain all $O(1)$ candidate chains in $\mathbb{E}(\nu)$ containing $(s, t)$. Let $C = (q, \ldots, r')$ be such a chain. To test whether $C$ is spanning in $F$, use Observation~\ref{obs:canonical} to find the unique regions $R(q)$, $R_b$, and $R(r')$.
Construct $\OCH(R(q) \cup R(r'))$ in $O(k \log k)$ time and perform a convex hull intersection test with $R_b$ in $O(\log m)$ time using the algorithm of Chazelle~\cite{chazelle1980detection}. The chain is spanning if and only if this test returns true.
\end{proof}

\begin{corollary}
        \label{cor:chain}
         We can dynamically maintain in $O(k^2 \log^2 m)$ time for each node $\nu \in T$ a balanced binary tree of all subchains of $\mathbb{E}^*(\nu)$ that are spanning in $F$.   
\end{corollary}

\begin{proof}
    Whenever the set of leaves of a node $\nu \in T$ changes (which occurs for $O(k \log m)$ nodes), invoke Lemma~\ref{lemm:chaincomputation} on $e(\nu)$ to maintain all chains that are spanning in $F$ (and contiguous subchains of $\CH(\nu)$ but not of $\CH(x)$ or $\CH(y)$). 
    The balanced binary tree associated to $\mathbb{E}^*(\nu)$  can be maintained in an identical manner as previous corollaries. 
\end{proof}



\begin{lemma} \label{lemm:hit_testing}
    For any node $\nu \in T$ and any edge $(x, y) \in \bE(\nu)$, we can find all contiguous subchains of $\mathbb{E}(\nu)$ that contain $(x, y)$ and are hit in $F$ in $O(k \log^2 m)$ time. 
\end{lemma}
\begin{proof}
Apply Lemma~\ref{lemm:candidate_chain_find} to find all candidate chains containing $(x, y)$. Then, for each candidate chain, consider the adjacent edge and apply the lemma again to construct all possible hit chains. For each resulting subchain, test whether the middle edge i$(s, t)$ is occupied as follows: 
 Construct $\band(R(s), R(t))$ in $O(k \log k)$ time~\cite{de2008computational}.
Temporarily remove $R(s)$ and $R(t)$ from $F$, and delete $s$ and $t$ from $V(F)$, in $O(k \log^2 m)$ time. Update the corresponding PHT $T'$ without updating auxiliary data. The hull $\CH(V(F) \setminus \{s,t\})$ is stored at the root of $T'$. By intersection testing between $\band(R(s), R(t))$ and $\CH(T')$ in $O(\log m)$ time~\cite{chazelle1980detection}, we can test if any hit chain is occupied.     
\end{proof}

\begin{corollary}
        \label{cor:occupied}
         We can dynamically maintain in $O(k^2 \log^3 m)$ time for each node $\nu \in T$ a balanced binary tree of all subchains of $\mathbb{E}^*(\nu)$ that are hit in $F$.   
\end{corollary}
\begin{proof}
As with Corollary~\ref{cor:chain}, we invoke \cref{lemm:hit_testing}  on $O(k \log m)$ nodes where the leaf set changes. For each, we update the associated trees using split and join operations.
\end{proof}

\begin{restatable}{theorem}{polygons}
    \label{lem:polygons_slow}
    Let $F$ be a family of $n$ simple polygons, where each region in $F$ has at most $k$ vertices. We can preprocess $F$ using $O(kn)$ space and $O(k^2 n \log^3 (kn) )$ time, such that given $P \sim F$ we reconstruct $\CH(P)$ using $O(kn)$ space and  $O( r(F, P) k^2 \log^3 (kn) )$ time.  
\end{restatable}

\begin{proof}
    By Corollaries~\ref{cor:canonical}, \ref{cor:dividing}, \ref{cor:chain}, and \ref{cor:occupied}, we may maintain our augmented Partial Hull Tree in $O(k^2 \log^3 (kn))$ time per retrieval. 
    This data structure maintains at its root four balanced trees storing all: non-canonical edges, canonical but non-dividing edges;
    subchains of $\OCH(F)$ that are spanning in $F$, or hit in $F$.
    We observe that Algorithm~\ref{alg:new_strategy} only considers occupied edges if all edges on $\OCH(F)$ are dividing.
    Then, any occupied edge corresponds to a hit subchain.
    Thus, we may immediately use these trees to execute Algorithm~\ref{alg:new_strategy}. This procedure performs $O( r(F, P))$ retrievals and so the theorem follows. 
\end{proof}

In the full version, we note that our approach has two bottlenecks. 
First, the algorithm from  Lemma~\ref{lemm:candidate_chain_find} may `skip' over $O(k)$ edges to find the largest candidate chain.
By applying a technique similar to skip-lists, we improve its running time by a factor $k$.
Secondly, our approach frequently uses a subroutine where for any two regions $R(s)$ and $R(t)$ we construct $\band(R(s), R(t))$ in $O(k \log k)$ time. 
We show that by storing for each region $R_i$ its convex hull, this construction time becomes polylogarithmic and we obtain the following:

\begin{restatable}{theorem}{polygonsImproved}\label{thm:polygons_improved}
    Let $F$ be a family of $n$ simple polygons where each region in $F$ has at most~$k$ vertices. We can preprocess $F$ using $O(kn \log n)$ space and $O(k n \log^3 (kn) )$ time, such that given $P \sim F$ we reconstruct $\CH(P)$ using  $O( r(F, P) \cdot k \log^3 (kn) )$ time. 
\end{restatable}

\section{\texorpdfstring{Improved running time for polygons with at most $k$ vertices }{Improved running time for polygons with at most k vertices}}\label{sec:improved_kgons}

Let $F$ be a family of $n$ polygons where each polygon has at most $k$ vertices. $V(F)$ denotes the set of all vertices in $F$ without duplicates and let $m = |V(F)|$. 
We note that our approach in Section~\ref{sec:kgons}  has two bottlenecks. 
First, the algorithm from  Lemma~\ref{lemm:candidate_chain_find} may `skip' over $O(k)$ edges to find the largest candidate chain.
By applying a technique similar to skip-lists, we improve its running time by a factor $k$.
Secondly, our approach frequently uses a subroutine where for any two regions $R(s)$ and $R(t)$ we construct $\band(R(s), R(t))$ in $O(k \log k)$ time. 
We show that by storing for each region $R_i$ its convex hull, this construction time becomes polylogarithmic

 To this end, recall the definition of a candidate chain:

\candidate*

\textcolor{white}{lipics is stupid} 

Lemma~\ref{lemm:candidate_chain_find} states that for any edge $(x, y)$ and any node $\nu$, we can find all candidate chains of $\mathbb{E}(\nu)$ that contain $(x, y)$ in $O(k)$ time. 
We show that instead of this expensive test, we may simply maintain for all nodes $\nu \in T$ all candidate chains of $\mathbb{E}^*(\nu)$ in a concatenable queue. 

Corollary~\ref{cor:chain} aims to dynamically maintain for a node $\nu$ all subchains of $\mathbb{E}^*(\nu)$ that are spanning in $F$. We showed that given a candidate chain, we can compute in $O(k)$ time whether it is spanning in $F$.
In this section we show that this test can be sped up to polylogarithmic time. 
Corollary~\ref{cor:occupied} aims to dynamically maintain for a node $\nu$ all chains that are hit in $F$. 
Given an edge $(s, t)$ (that is shared between two candidate chains) we tested in $O(k \log^2 n)$ time whether $(s, t)$ was occupied in $F$.
We show that we can perform this test in polylogarithmic time instead.

\subsection{A new data structure}
We denote for any family of regions $F$ by $\PHT(F)$ a Partial Hull Tree on $V(F)$. 
The base of our data structure is $T = \PHT(F)$.
We augment $T$ by maintaining for each $\nu \in \PHT(F)$  four balanced trees associated with~$\mathbb{E}^*(\nu)$, where edges are ordered by their appearance in $\mathbb{E}^*(
\nu)$:

\begin{enumerate}
    \item a tree storing all edges $(s, t) \in \mathbb{E}^*(\nu)$ that are non-canonical in $F$,
    \item a tree storing all edges $(s, t) \in \mathbb{E}^*(\nu)$ that are canonical and non-dividing in $F$, 
        \item a tree storing all contiguous subchains of $\mathbb{E}^*(\nu)$ that are spanning in $F$,
      \item a tree storing all contiguous subchains of $\mathbb{E}^*(\nu)$ that are hit in $F$.
\end{enumerate}

\noindent
We assign to each region in $F$ a unique ID that is the bit-string of an integer in $[n]$. 
Whenever we retrieve a region in $F$, we assign to the retrieved point the same ID as the original region. 
For any pair of integers $(i, j) \in [\log n] \times \{ 0, 1 \}$ denote by $F(i, j)$ all regions in $F$ whose ID has at position $i$ the bit $j$. We store $O(\log n)$ additional Partial Hull Trees.
These are $\PHT(F(i, j))$ for all $(i, j) \in [\log n] \times \{ 0, 1 \}$.
Note that we do not augment these PHTs. 
Since a PHT over $m$ points has $O(\log^2 m)$ worst-case update time, and each retrieval in $F$ triggers $O(k)$ updates, we may immediately note the following:

\begin{theorem}
    \label{thm:PHT_copies}
Let $F$ be a family of $n$ polygons where each polygon has at most $k$ vertices. We may maintain for all $(i, j) \in [\log n] \times \{ 0, 1 \}$ the tree $\PHT(F(i, j))$ subject to retrievals in~$F$ using $O(kn \log n)$ space and $O(k \log^3 (kn))$ worst-case retrieval time. 
\end{theorem}

\noindent
During preprocessing, we store additionally for each region $R \in F$ the edges of the full convex hull of $R$ in their cyclical order. 
We denote this convex hull by $H(R)$. Whenever we retrieve a region $R_i$ and obtain $p_i$, we set $H(R) = \{ p_i \}$. 
Finally, we define a new concept of candidate edges (which we will use to efficiently find the candidate chains of a node $\nu \in T$). 

\begin{definition}
    Any edge $(s, t)$ is a \emph{candidate edge} of a region $R_b \in F$ if $s, t \in V(R_b)$ and $s, t \not \in V(R_x)$ for $x \neq b$. 
\end{definition}

\begin{definition}
    Let $C$ be a convex chain of edges. The \emph{compressed candidate chain} of $C$ is a sequence of objects, such that each candidate edge $(s, t)$ (of a region $R_b$) in $C$ is stored in a unique object together with all edges on the maximal contiguous subchain of $\OCH(\nu)$ containing $(s, t)$ consisting of candidate edges of $R_b$ only.
\end{definition}

We augment $T$ once more by maintaining for each $\nu \in \PHT(F)$ a family of balanced binary trees. Specifically, we store:

\begin{enumerate}
    \setcounter{enumi}{4}
    \item the compressed candidate chain $\Gamma^*(\nu)$ of $\mathbb{E}^*(\nu)$ in a balanced binary tree (ordered along~$\mathbb{E}^*(\nu))$. Each leaf of this tree stores another balanced binary tree of candidate edges stored in the object (ordered along $\mathbb{E}^*(\nu)$). 
\end{enumerate}

\begin{corollary}
    We can dynamically maintain in $O(k \log^2 m)$ time for each node $\nu \in T$ the above data structure. 
\end{corollary}

\begin{proof}
    Any candidate edge is canonical in $F$. In particular, we can use Observation~\ref{obs:canonical}
    to test whether any bridge $(s, t)$ in $T$ is a candidate edge, and of which region. 
    Thus, we may efficiently maintain the set of all bridges in $T$ that are candidate edges of some region in $F$. What remains is to recompute the corresponding binary trees during concatenable queue operations.
    We use a strategy identical to the maintenance of our other augmenting trees. 

    Formally, let $\nu \in T$ have children $x$ and $y$ and let $\Gamma(\nu)$ denote the compressed candidate chain of $\mathbb{E}(\nu)$ in a balanced binary tree. Let each leaf of this tree store another balanced binary tree of candidate edges stored in the object. 
    Given $(\Gamma(\nu), e(\nu), \Gamma^*(x))$ we may compute $\Gamma(x)$ by splitting $\Gamma(\nu)$ at $e(\nu)$ and joining the result with $\Gamma^*(x)$. 
    
    Given $(\Gamma(x), e(\nu), \Gamma(y))$ we can also compute $(\Gamma(\nu), \Gamma^*(x), \Gamma^*(y))$. Indeed, we may first split $\Gamma(x)$ at $e(\nu)$ to obtain $\Gamma^*(x)$ and $C_x$ (the prefix of $\Gamma(\nu)$ until $e(\nu)$). Let $e(\nu)$ be a candidate edge of a region $R_b$. 
    We insert $e(\nu)$ into $C_x$ in $O(\log m)$ time (creating either a new rightmost leaf, or adding it to the rightmost leaf of $C_x$ whenever $e(\nu)$ is a candidate edge of the same region $R_b$ as edges in that leaf. 
    Similarly, we obtain $\Gamma^*(y)$ and $C_y$.
    We then join $C_x$ and $C_y$. Let the rightmost leaf of $C_x$ contain candidate edges of some region $R_b$ and the leftmost leaf of $C_y$ contain candidate edges of some region $R_c$. If $b = c$ we join these two leafs in $O(\log m)$ time. 
   
    By applying this top-down and bottom-up traversal in the same manner as we maintain concatenable queues, the corollary follows. 
\end{proof}

\subsection{Speeding up our reconstruction algorithm}

We show that this new structure significantly speeds up Lemma~\ref{lemm:candidate_chain_find}:

\begin{lemma} \label{lemm:candidate_chain_find_fast}
    Let $\nu \in T$.
    There are at most two candidate chains that contain an edge $(s, t)$ of $\mathbb{E}(\nu)$ and, given $\mathbb{E}(\nu)$ and our data structure, we can find these in $O(\log m)$ time.
\end{lemma}
\begin{proof}
    The proof is identical to that of \cref{lemm:candidate_chain_find}, with one change: Instead of traversing $\bE(\nu)$ for $O(k)$ steps, we do a lookup in $\Gamma(\nu)$ to find the compressed candidate chain that contains $(s, t)$. Let $e$ and $e'$ be the first and last candidate chain in this object. Then, all edges in between $e$ and $e'$ are candidate edges and hence dividing, so we only need to find the predecessor of $e$ and successor of $e'$ in $\bE(\nu)$, and check that these edges each are dividing and have endpoints in two different regions.
\end{proof}

\begin{lemma}
    \label{lemm:occupied_fast}
    For any $\nu \in T$ and any edge $(s, t) \in \bE(\nu)$ that is dividing in $F$, we can test whether $(s, t)$ is occupied in $F$ in $O(\log^2 m)$ time. 
\end{lemma}

\begin{proof}
    By Observation~\ref{obs:canonical}, we obtain two unique areas $R(s), R(t)$ with $s \in V(R(s))$ and $t \in V(R(t))$ in constant time. The edge $(s, t)$ is occupied if and only if $\band(R(s), R(t))$ contains a vertex in $V(F - R(s) - R(t)) - \{s\} - \{t\}$. 
    Given $H(R(s))$ and $H(R(t))$, we obtain the edges of $\band(R(s), R(t))$ as a balanced binary tree in cyclical ordering in $O(\log k)$ time through the bridge-finding algorithm~\cite{DBLP:conf/stoc/OvermarsV80}.
    Since $R(s)$ and $R(t)$ are vertically separated, we may split $\band(R(s), R(t))$ into two convex areas: $C(s)$ containing $H(R(s))$ and $C(t)$ containing $H(R(s))$.
    The edge $(s, t)$ is occupied if and only if one of two conditions holds:
    
    \begin{enumerate}
        \item $C(s)$ contains a vertex in $V(F - R(s) ) - \{s \}$, or
        \item $C(t)$ contains a vertex in $V(F - R(t))  - \{ t \}$.
    \end{enumerate}

\noindent
    We show how to test for the first condition, testing the second condition is analogous. 
    In the special case that $R(s)$ is a point region we simply perform a deletion on the leaf storing $\{ s \}$ in $O(\log^2 m)$ time (we do not update our auxiliary data structures).
    Let $\rho$ be the root of the updated PHT $T$. 
    We do convex hull intersection testing between $C(s)$ and $\mathbb{E}(\rho)$ to find if $C(s)$ contains any vertex in $V(F) - \{ s \}$. 

    If $R(s)$ is not a point region then $s$ is uniquely a vertex of $R(s)$. What remains is to test whether $C(s)$ contains a vertex in $V(F - R(s) )$. 
    Let $\sigma$ be the bit-string corresponding to the ID of $R(s)$ and let $\sigma[i]$ denote inverse of the $i$'th bit in $\sigma$.
    Note that for all $i \in [\log n]$, $R(s) \not \in F(i, \sigma[i])$ and that for all $R_b \neq R(s)$ there exists an integer $i$ such that $R_b \in F(i, \sigma[i])$. 
    
    For all $i \in [\log n]$, we consider the root $\rho$ of $\PHT(F(i, \sigma[i]))$ and do convex hull intersection testing in $O(\log m)$ time between $C(s)$ and $\mathbb{E}(\rho)$ to determine whether there is a vertex in $V(F(i, \sigma[i]))$ that is contained in $C(s)$. 
    It follows that in $O(\log^2 m)$ time we can test whether $C(s)$ contains a vertex in $V(F - R(s) ) - \{s \}$ and so the lemma follows.    
\end{proof}

\begin{lemma} \label{lemm:hit_testing_fast}
    For any node $\nu \in T$ and any edge $(x, y) \in \bE(\nu)$, we can find all contiguous subchains of $\mathbb{E}(\nu)$ that contain $(x, y)$ and are hit in $F$ in $O(\log^2 m)$ time. 
\end{lemma}
\begin{proof}
    The proof is identical to \cref{lemm:hit_testing}, using \cref{lemm:candidate_chain_find_fast,lemm:occupied_fast}.
\end{proof}

\begin{corollary} \label{cor:occupied_fast}
     We can dynamically maintain in $O(k \log^3 m)$ time for each node $\nu \in T$ a balanced binary tree of all subchains of $\mathbb{E}^*(\nu)$ that are hit in $F$.   
\end{corollary}
\begin{proof}
    The proof is identical to \cref{cor:occupied}, using \cref{lemm:hit_testing_fast}.
\end{proof}

What remains is to find for each $\nu \in T$ all spanning subchains of $\mathbb{E}^*(\nu)$. 

\begin{lemma}
    \label{lemm:chaincomputation_fast}
 For any $\nu \in T$, given $e(\nu) = (s, t)$ and 
     $\mathbb{E}(\nu)$, we may find the at most two contiguous subchains $C$ of $\mathbb{E}(\nu)$ with $s, t \in C$ that are spanning in $F$ in $O(\log^2 m)$ time.
\end{lemma}

\begin{proof}

We invoke \cref{lemm:candidate_chain_find_fast} to find the at most two contiguous subchains $C$ of $\mathbb{E}(\nu)$ with $s, t \in C$,
Given any candidate chain $C = (q, \ldots, s, \ldots, r)$, we show how to test in $O(\log m)$ time whether it is spanning in $F$. 
We apply Observation~\ref{obs:canonical} to obtain three unique areas $R(q)$, $R(s)$ and $R(r)$. 
The chain is spanning if and only if $R(s)$ intersects the inside of $\OCH(R(q), R(r))$. 
Since $R(q)$ and $R(r)$ are vertically separated and we store their convex hulls in a data structure, we may obtain $\OCH(R(q) \cup R(r))$ as a balanced binary tree in $O(\log m)$ time.
We then do convex hull intersection testing between the convex hull of $R(s)$ and $\OCH(R(q) \cup R(r))$ in $O(\log m)$ time.

\end{proof}

\begin{corollary}
        \label{cor:chain_fast}
         We can dynamically maintain in $O(k \log^3 m)$ time for each node $\nu \in T$ a balanced binary tree of all subchains of $\mathbb{E}^*(\nu)$ that are spanning in $F$.   
\end{corollary}

\begin{proof}
       This proof is identical to that of \cref{cor:chain}, invoking \cref{lemm:chaincomputation_fast} instead of \cref{lemm:chaincomputation}.
\end{proof}

\polygonsImproved*

\begin{proof}
    By Corollaries~\ref{cor:canonical}, \ref{cor:dividing}, \ref{cor:occupied_fast}, and \ref{cor:chain_fast}, we may maintain our augmented Partial Hull Tree in $O(k \log^3 (kn))$ time per retrieval. 
    This data structure maintains at its root four balanced trees storing all: non-canonical edges, canonical but non-dividing edges;
    subchains of $\OCH(F)$ that are hit in $F$, or spanning in $F$.
    We observe that Algorithm~\ref{alg:new_strategy} only considers occupied edges if all edges on $\OCH(F)$ are dividing.
    Then, any occupied edge corresponds to a hit subchain.
    We may immediately use these lists to execute Algorithm~\ref{alg:new_strategy}. This procedure performs $O( r(F, P))$ retrievals and the theorem follows.    
\end{proof}

\section{\texorpdfstring{Pairwise disjoint disks with a radius in $[1, k]$, or ply at most $k$}{Pairwise disjoint disks of bounded radius}}\label{sec:disks}

We simultaneously consider two scenarios.
We let either $F$ be a family of pairwise disjoint disks with radii in $[1, k]$,  or $F$ be a set of unit disks of ply $k$. 
We preprocess $F$ and, given access to some $P \sim F$, we reconstruct $\CH(P)$ in $O(r(F, P) \cdot k \log^3 n )$ time. 

\subparagraph{Core idea and approach.}
We cannot apply the approach of Section~\ref{sec:kgons}, since we no longer have a finite vertex set whose convex hull equals the convex hull of $F$.
Instead, for any family~$F$ of disks and point regions, let $V(F)$ denote the set of vertices obtained by including for each region $R \in F$ its centre point (we assume that all centre points lie in general position).
Our key idea is to maintain a balanced decomposition of vertical lines through $V(F)$. 

Each node $\nu$ in the decomposition has a corresponding line $\ell$ that has at most a constant fraction of the remaining regions on either side.
In particular, the line $\ell$ partitions the remaining regions into three families: disks strictly to the left of $\ell$, disks strictly to the right of $\ell$, and disks that intersect $\ell$.
We observe that, around any vertical line $\ell$, the upper quarter convex hull cannot be overly complex.
We prove this separately for our two cases of $F$. The remainder of this section then relies on this observation.

\begin{lemma}
\label{lemm:hull_complexity}
     For $F$ a family of pairwise disjoint disks with radii in $[1, k]$,  $A \subseteq F$, and $S$ a vertical slab of width $2k$, let $F' = \{ R \in F \retrieve A \mid R \subset S \}$. 
     Then $\OCH(F')$ has $O(k)$ edges. 
\end{lemma}
\begin{proof}
    For each $R' \in F'$ that appears on $\OCH(F')$ we choose an arbitrary vertex $r' \in V(R') \cap \OCH(F')$. 
    This gives a (sub-) set $V = (v_1, \ldots, v_m)$ of the vertices of $\OCH(F')$, ordered along $\OCH(F')$. 
    We first show that $|V| \in O(k)$. 
    Given $r' \in V$ with $r' \in R'$, let $R$ be the corresponding region in $F$.
    We fix an arbitrary unit disk $D_{r'} \subseteq R$ that contains $r'$. 
    Doing this for all vertices in $V$  gives a set $D$ of $|V|$ pairwise disjoint unit disks. 
    We now claim that for any $i \in [m - 6]$, the subcurve of $\OCH(F')$ between $v_i$ and $v_{i+6}$ has length at least two.
    Indeed, if this subcurve has length less than two, then it lies inside a circle of radius one.
    However, seven pairwise disjoint unit disks cannot all intersect the same unit circle.
    Since $\OCH(F')$ is the \emph{upper quarter} outer convex hull of $F'$, the length of $\OCH(F') \cap S$ is $\Theta(k)$. So $|V| \in O(k)$, and thus the number of distinct regions in $F'$ on $\OCH(F')$ is in $O(k)$. 

    It may be that a region in $F'$ appears multiple times on $\OCH(F')$.
    Let $\sigma$ denote the sequence of regions that appear on $\OCH(F')$ from left to right. Since $F$ is a family of pairwise disjoint disks, this sequence cannot contain the subsequence $R_a, R_b, R_a, R_b$ for any distinct $a, b \in [n]$. 
    It follows that the total number of edges on $\OCH(F')$ is $O(k)$.
\end{proof}

\begin{lemma}
\label{lemm:hull_complexity_overlap}
     For $F$ a family of unit-size disks of ply $k$, let  $A \subseteq F$, and $S$ be a vertical slab of width $2$, let $F' = \{ R \in F \retrieve A \mid R \subset S \}$. 
     Then $\OCH(F')$ has $O(k)$ edges. 
\end{lemma}

\begin{proof}
    Let $z$ be the lowest line such that all regions in $F \cap S$ lie below $z$. 
    Let $Z$ be the subset of $S$ within distance $2$ of $z$. If  
    $p$ is a point in $S$ that appears on $\OCH(F')$ then, because we are computing the upper-quarter hull, $p$ must lie in $Z$. 
        Since all regions have unit size, and all points in the plane intersect at most $k$ regions in $F$, it follows that there are at most $O(k)$ regions $R'$ for which their corresponding point appears on $\OCH(F')$.    
\end{proof}

All remaining proofs work both when $F$ is a family of pairwise disjoint disks with radii in $[1, k]$, and when $F$ is a family of unit-size disks where each point in the plane intersects at most $k$ disks:

\begin{corollary}
\label{corr:hull_complexity}
    If $\ell$ is a vertical line and $F' = \{ R \in F \mid R \cap \ell \neq \emptyset \}$, then $|\OCH(F')| \in O(k)$. 
\end{corollary}

\subparagraph{A new data structure.} \cref{corr:hull_complexity} leads to the following data structure.

\begin{definition}
    For a family $F$ of disks and point regions, we define the \emph{Median Cut Decomposition} (MCD) of $F$ as a recursive hierarchical decomposition $\MCD(F)$:
    \begin{itemize}
        \item $\nu$ is the root; a vertical line $\ell(\nu)$ whose $x$-coordinate is approximately the median of $V(F)$. 
               \item $\nu$ stores $\phi(\nu)$: all regions in $F(\nu)$ that intersect $\ell(\nu)$ ordered from top to bottom. 
        \item $F(\nu)$ denotes the input regions $F$.
        \item $L(\nu)$ denotes all regions in $F(\nu)$ strictly left of $\ell(\nu)$. The left child of $\nu$ is $\MCD(L(\nu))$. 
        \item $R(\nu)$ denotes all regions in $F(\nu)$ strictly right of $\ell(\nu)$. The right child of $\nu$ is $\MCD(R(\nu))$.
    \end{itemize}
\end{definition}

\begin{lemma}
    Let $F$ be a family of $n$ disk and point regions. We can construct $\MCD(F)$ in $O(n \log n)$ time. 
\end{lemma}

\noindent
We dynamically maintain $\MCD(F)$ subject to retrievals in $F$. We then show that we can augment the Median Cut Decomposition so that each node $\nu \in \MCD(F)$ also maintains a \emph{concatenable queue} of $\OCH(F(\nu))$.
Finally, we use this data structure to obtain an instance-optimal reconstruction strategy in a manner analogous to \cref{sec:kgons}.

\subsection{\texorpdfstring{Dynamically maintaining $\MCD(F)$}{Dynamically maintaining MCD(F)}}

A retrieval operation selects a disk region $R_i$ in $F$ and replaces it with a point region $p_i$.
There exists some $\nu \in \MCD(F)$ where $R_i \in \phi(\nu)$. To perform the retrieval, we first delete~$R_i$ from $\phi(\nu)$ in $O(\log n)$ time. In the special case that $p_i$ lies on $\ell$, we insert $p_i$ into $\phi(\nu)$ in $O(\log n)$ time. Else, we recursively insert $p_i$ into $\MCD(L(\nu))$, or $\MCD(R(\nu))$, depending on whether $p_i$ is left or right of $\ell(\nu)$.

We want to maintain the property that $\MCD(F)$ is a balanced binary tree.
Many schemes for balancing binary trees rely on tree rotations.
However, when rotating around a node $\nu$ (making it the parent node of its parent) the family of regions $F(\nu)$ (and thereby $\phi(\nu)$) may have $O(n)$ recourse.
To avoid this issue, we use \emph{scapegoating}, introduced by Andersson~\cite{andersson1989improving}.
Instead of using rotations to maintain height-balance, scapegoating rebuilds the balanced tree on $\nu$ from scratch whenever the height of a node $\nu$ becomes too large compared to the size of $\nu$.
If constructing the tree takes $O(n \cdot f(n))$ time, then insertions under scapegoating have $O(f(n)  \log n )$ amortized update cost. 
We briefly note that when constructing $\MCD(F)$ as a perfectly balanced tree, the amortised potential in each node is zero. This implies:

\begin{theorem}
    Let $F$ be a family of $n$ disks and point regions. We can construct $\MCD(F)$ in $O(n \log n)$ time such that we can dynamically maintain $\MCD(F)$, subject to retrievals in $F$, in $O(\log^2 n)$ amortized update time. 
\end{theorem}

\subsection{ \texorpdfstring{ Augmenting $\MCD(F)$}{Augmenting MCD(F)}}

We augment $\MCD(F)$ by adding several data structures to each node. 
First and foremost, we note that all regions in $\phi(\nu)$ intersect the same vertical line $\ell(\nu)$. Therefore, we may immediately store these regions in a Partial Hull Tree (where regions are ordered along $\ell(\nu)$) to dynamically maintain $\OCH(\phi(\nu))$.
A PHT has $O(\log^2 n)$ update time and $O(n \log n)$ construction time, so we may maintain this data structure by adding $O(\log^2 n)$ amortised update time to each update. 
We then add  \emph{concatenable queues}:

\begin{definition}
    Let $\nu \in \MCD(F)$ have children $x$ and $y$. We define $\mathbb{E}(\nu)$ as a balanced tree over the edges of $\OCH(F(\nu))$.  A \emph{bridge} of $\nu$ is any edge in $\mathbb{E}(\nu)$ that is not in $\mathbb{E}(x)$ or $\mathbb{E}(y)$. 
We define a \emph{concatenable queue} $\mathbb{E}^*(\nu)$ as a balanced tree over all edges in $\mathbb{E}(\nu)$ that are not in $\mathbb{E}(\gamma)$ for some ancestor $\gamma$ of $\nu$. 
\end{definition}

\begin{definition} \label{def:subchains}
    Let $\nu \in \MCD(F)$ have children $x$ and $y$ and let $S$ be the vertical slab of width $2k$ around $\ell(\nu)$ (if $F$ are overlapping unit disks, we use width $2$ instead). We define $C_x$, $C_y$ and $C^*_\nu$ as the contiguous subchains of $\mathbb{E}(x)$, $\mathbb{E}(y)$ and $\OCH(\phi(\nu))$ that intersect $S$. 
\end{definition}

\noindent
We store for each node $\nu$ with children $x$ and $y$:
\begin{itemize}
    \item the concatenable queue $\mathbb{E}^*(\nu)$, and
    \item the chains $C_x$ and $C_y$ as a balanced binary tree.
\end{itemize}

\noindent
This data structure uses $O(n \log n)$ space since any bridge in $\MCD(F)$ can be stored $O(\log n)$ times across $C_x$ for $x \in \MCD(F)$. Next we show that we can dynamically maintain it.

\begin{lemma}
    \label{lemm:bottomup}
    Let $\nu \in \MCD(F)$ have children $x$ and $y$. 
    Given $\mathbb{E}(x)$, $\mathbb{E}(y)$ and $\OCH(\phi(\nu))$, we may compute $(\mathbb{E}(\nu), \mathbb{E}^*(x), \mathbb{E}^*(y), C_x, C_y)$ in $O(\log n + k)$ time.
\end{lemma}

\begin{proof}
    We consider $C_x, C_y$ and $C_\nu^*$ from Definition~\ref{def:subchains}. 
    By \cref{lemm:hull_complexity} (or, Lemma~\ref{lemm:hull_complexity_overlap}), $C_x, C_y$ and $C_\nu$ have $O(k)$ edges, which we obtain in $O(k + \log n)$ time.      
We store a copy of $C_x$ and $C_y$ as a balanced binary tree. 
    We then apply merge sort and get $C^* = \OCH(C_x \cup C_y \cup C_\nu^*)$ as a balanced binary tree in $O(k)$ time.

    The chains $\mathbb{E}(x) - C_x$, $C^*$, and $\mathbb{E}(y) - C_y$ are three $x$-separated chains with at most $n$ vertices each.
    Thus, we may pairwise apply the bridge-finding algorithm for $x$-separated convex hulls to find the at most two bridges between these three hulls. 
    Given the two bridges, we split $\mathbb{E}(x)$ and $\mathbb{E}(y)$ at the respective bridges into two trees to obtain $\mathbb{E}^*(x)$ and $\mathbb{E}^*(y)$. 
    We then join the remaining $O(1)$ trees in $O(\log n)$ time to get $\mathbb{E}(\nu)$.     
\end{proof}

\begin{lemma}
    \label{lemm:top_down}
    Let $\nu \in \MCD(F)$ have children $x$ and $y$. Given $(\mathbb{E}(\nu), \mathbb{E}^*(x), \mathbb{E}^*(y), C_x, C_y)$, we may compute $\mathbb{E}(x)$ and $\mathbb{E}(y)$ in $O(\log n)$ time. 
\end{lemma}

\begin{proof}
     Let $S$ denote the vertical slab    around $\ell(\nu)$ of Definition~\ref{def:subchains}. 
     We split $\mathbb{E}(\nu)$ into three trees, containing edges that strictly precede $S$, intersect $S$, and strictly succeed $S$ in $O(\log n)$ time. 
     Let $L$ be the edges strictly preceding $S$.
     We split $L$ on $e(\nu)$ in $O(\log n)$ time to obtain the subset $L'$ of $L$ that is in $\mathbb{E}(\nu)$.
    We then combine $L'$, $\mathbb{E}^*(x)$ and $C_x$ in logarithmic time through the tree join operation to get $\mathbb{E}(x)$. The tree $\mathbb{E}(y)$ is obtained analogously. 
\end{proof}

\noindent
These lemmas and observations allow us to define and maintain concatenable queues.

\begin{lemma}
    \label{lemm:dynamic_c_queue}
    We can dynamically maintain our augmented $\MCD(F)$ subject to retrievals in $F$, using $O(n \log n)$ space and $O( k \log n + \log^2 n)$ amortized time.
\end{lemma}

\begin{proof}
    Our update algorithm is similar to the algorithm that maintains a PHT, which first does top-down traversal followed by a bottom-up construction. We wish to maintain the invariant that $\MCD(F)$ is a balanced tree and that each node $\nu \in \MCD(F)$  with children $x$ and $y$ stores $\mathbb{E}^*(\nu)$, $\OCH(\phi(\nu))$ and $C_x$ and $C_y$. 
    Consider retrieving a region $R_i \in F$ which gives the point $p_i$. Let $\nu^*$ denote the unique node in $\MCD(F)$ where $R_i \in \phi(\nu^*)$. 
    In the special case where $p_i$ lies on $\ell(\nu^*)$, we update the PHT storing $\OCH(\psi(\nu))$ in $O(\log^2 n)$ time.
    Otherwise, we insert a new leaf $\lambda$ in the subtree rooted at $\nu^*$.

    \cref{lemm:bottomup} implies an $O( n \cdot ( k + \log n) )$ time construction algorithm for our augmented $\MCD(F)$.
    Thus, after inserting $\lambda$, we may use the scapegoating  amortised rebuilding technique to maintain the balance of $\MCD(F)$ in $O(k \log n + \log^2 n)$ amortised time. 
    Suppose our amortized rebuilding scheme rebuilds some subtree $\mu^*$ that contains $\lambda$ (if no subtrees were rebuilt, $\mu^*$ is the parent of $\lambda$).  Let $\pi$ denote the path in $\MCD(F)$ from the root to $\mu^*$. What remains is to restore $\MCD(F)$ along $\pi$. 

    We invoke \cref{lemm:top_down} over the top-down traversal of $\pi$.
    This way, we obtain for each $\nu \in \pi$,  with children $(\mu, \gamma)$, the trees $\mathbb{E}(\mu)$ and $\mathbb{E}(\gamma)$  before the update. If $\lambda$ is to be in the subtree rooted at $\mu$, we keep $\mathbb{E}(\gamma)$  in memory and invoke \cref{lemm:bottomup} on $\mu$ next.    
    We do this until we arrive at $\mu^*$, which has one new child and one unchanged child node, and we start a bottom-up traversal. 
    Let $\nu \in \pi$ be the next node in our bottom-up traversal and let $\nu$ have children $x$ and $y$. 
    Given the up-to-date trees $\mathbb{E}(x)$, $\mathbb{E}(y)$ and $\CH(\phi(\nu))$ we invoke \cref{lemm:bottomup} to compute $(\mathbb{E}(\nu), \mathbb{E}^*(x), \mathbb{E}^*(y))$ in $O(k + \log n)$ time. Since $\pi$ as length at most $O(\log n)$, the remaining update takes $O(k \log n + \log^2 n)$ time.
\end{proof}

\subsection{Concluding the argument.}
From here-on, we use our data structure in an identical manner to \cref{sec:kgons}. Let $\nu \in \MCD(F)$. We augment the concatenable queues of a node $\nu$ by storing four additional balanced trees:

\begin{enumerate}
    \item a balanced binary tree storing all edges $(s,t) \in \mathbb{E}^*(\nu)$ that are non-canonical in $F$,
    \item a balanced tree storing all edges $(s, t) \in \mathbb{E}^*(\nu)$ that are canonical and non-dividing in $F$, 
       \item a balanced tree storing  all contiguous subchains of $\mathbb{E}^*(\nu)$ that are spanning in $F$, and
    \item a balanced tree storing  all contiguous subchains of $\mathbb{E}^*(\nu)$ that are hit in $F$. 
\end{enumerate}

\begin{lemma}
    \label{lemm:candidate_disk}
     For any $\nu \in T$, there are $O(k)$ subchains of $\CH(F(\nu))$ that contain a bridge of $\nu$ and that are candidate chains.  Given $\mathbb{E}(\nu)$ and the Booleans in our data structure, we can find these in $O(k)$ time.
\end{lemma}

\begin{proof}
 Since regions in $F$ are pairwise disjoint disks, or 
 unit-size disks, any candidate chain has at most three edges.  
    Given the set of bridges of $\nu$ and $\mathbb{E}(\nu)$ we find for each bridge the preceding and succeeding edge in constant additional time (we do this by maintaining for each edge in $\mathbb{E}(\nu)$ pointers to its predecessor and successor). Recall that Lemma~\ref{lemm:hull_complexity} and \ref{lemm:hull_complexity_overlap} implies that there are only $O(k)$ of these bridges. 
\end{proof}

\noindent
What follows are three corollaries that we can maintain our augmentations to $\MCD(F)$ in $O(k \log^3 n)$ amortised time.

\begin{corollary}   \label{corr:dividing_disk}
         We can dynamically maintain in $O(k \log n + \log^2 n)$ amortised time for each node $\nu \in T$ a balanced binary tree on all edges $(s, t) \in \mathbb{E}^*(\nu)$ that are non-dividing in $F$.
\end{corollary}

\begin{proof}
   We can maintain $\MCD(F)$ in amortised $O(k \log n + \log^2 n)$ time and we can test whether any newly created bridge is dividing in constant time.     
\end{proof}

\begin{lemma}
    \label{lemm:disk_chain_test}
    We can test if a candidate chain $C$ is spanning in $F$ in constant time. 
\end{lemma}

\begin{proof}
        The chain $C$ corresponds to three regions $R_a, R_i, R_b \in F$ which we obtain in $O(1)$ time. We then test whether $R_i$ intersects the convex hull of $(R_a, R_b)$ in constant time.    
\end{proof}

\begin{corollary}
    \label{corr:subchains_disk}
        We can maintain for each $\nu \in \MCD(F)$ all subchains chains of $\mathbb{E}^*(\nu)$ that are spanning in $F$ in a concatenable queue in $O( k \log n)$ amortized time.
\end{corollary}

\begin{proof}
    For an update at a node $\mu$, only the $O(\log n)$ nodes on the root-to-leaf path of $\mu$ can gain or lose a contiguous subchain of $\mathbb{E}(\nu)$ that is spanning in $F$.
    For each such node $\nu$ with children $x, y$, we invoke \cref{lemm:candidate_disk} to find all $O(k)$ candidate chains that are contiguous subchains of $\CH(\nu)$ but not of $\CH(x)$ or $\CH(y)$.
    We use \cref{lemm:disk_chain_test} to check which of these chains are spanning in $F$.
    We then maintain the balanced binary tree associated to each concatenable queue $\mathbb{E}^*(\nu)$  in an identical manner as previous corollaries.
    This takes $O(k \log n)$ time per update. During scapegoating, we consider all $O(n)$ nodes $\nu$, so this takes amortized $O(k \log n)$ time too. 
\end{proof}

\begin{lemma} \label{lemm:occupied_disk}
     For any $\nu \in \MCD(F)$, for any $(s, t) \in \mathbb{E}(\nu)$, we can test whether $(s, t)$ is occupied in $F$ in $O(\log^2 n)$ time.
\end{lemma}

\begin{proof}
        Let $s \in R_a$ and $t \in R_b$.
        Let $\pi_a$ denote the path from the root to the node $\mu$ where $R_a \in \phi(\mu)$. 
        Let $L_a$ be the set of all outer convex hulls of nodes in $\pi_a$ and their children, i.e. $L_a = \{ \CH(F(\gamma)) \mid  \gamma \textnormal{ child of a node } \gamma' \in \pi_a  \}$.  
        Note that $L_a$ contains $O(\log n)$ convex hulls, and that we can obtain these convex hulls in $O(\log^2 n)$ time by invoking Lemma~\ref{lemm:top_down} along $\pi_a$.
        Let $\Phi_a = \{ \CH(\phi(\gamma)) \mid \gamma \in \pi_a \}$.  Then $|\Phi_a| \in O(\log n)$ and we obtain it through a top-down traversal of $\pi_a$. 
        Finally we define $H_a = \OCH(\phi(\mu) - R_a - R_b )$ which we obtain in $O(\log^2 n)$ by updating the PHT over $\OCH(\phi(\nu))$. We define $L_b$, $\Phi_b$ and $H_b$ analogously and also obtain it in $O(\log^2 n)$ time. 
        
        Any vertex of $\OCH(F - R_a - R_b)$ must also be a vertex that appears in a chain in $\bigcup \{  L_a, L_b, \Psi_a, \Psi_b, H_a, H_b \}$. 
        If follows that $(s, t)$ is occupied if and only if the convex hull of $(R_a, R_b)$ contains a vertex in a chain in $\bigcup \{  L_a, L_b, \Psi_a, \Psi_b, H_a, H_b \}$.
        There are $O(\log n)$ convex chains in this set.
        For each of these chains we use the standard algorithm to test whether two convex hulls intersect in $O(\log n)$ time~\cite{chazelle1980detection}. 
\end{proof}

\begin{lemma}
        \label{lemm:hit_disk}
     For any $\nu \in T$, there are $O(k)$ subchains of $\CH(F(\nu))$ that contain a bridge of $\nu$ and that are hit in $F$.
     Given $\mathbb{E}(\nu)$ and the Booleans in our data structure, we can find these in $O(k \log^2 n)$ time.
\end{lemma}
\begin{proof}
    We invoke \cref{lemm:candidate_disk} to find all $O(k)$ candidate chains that contain a bridge of $\nu$.
    We extend each such chain by one edge to the left or right to a chain $C'$, and use \cref{lemm:occupied_disk} to check whether $C'$ is hit in $F$.
\end{proof}

\begin{corollary} \label{corr:hit_disk}
            We can maintain for each $\nu \in \MCD(F)$ all subchains chains of $\mathbb{E}^*(\nu)$ that are hit in $F$ in a concatenable queue in $O( k \log^3 n)$ amortized time.
\end{corollary}

\begin{proof}
    For an update at a node $\mu$, only the $O(\log n)$ nodes on the root-to-leaf path of $\mu$ can gain or lose a contiguous subchain of $\mathbb{E}(\nu)$ that is hit in $F$.
    For each such node $\nu$ with children $x, y$, we invoke \cref{lemm:hit_disk} to find all $O(k)$ chains that are hit in $F$ and that are contiguous subchains of $\CH(\nu)$ but not of $\CH(x)$ or $\CH(y)$.
    We then maintain the balanced binary tree associated to each concatenable queue $\mathbb{E}^*(\nu)$  in an identical manner as previous corollaries.
    This takes $O(k \log^3 n)$ time per update. During scapegoating, we consider all $O(n)$ nodes $\nu$, so this takes amortized $O(k \log^3 n)$ time too. 
\end{proof}

Through \cref{corr:dividing_disk}, \ref{corr:hit_disk}, and \ref{corr:subchains_disk} we always have at the root of $\MCD(F)$ a doubly linked list of all edges of $\OCH(F)$ that are non-dividing,
and a doubly linked list of all continuous subchains of $\OCH(F)$ that are spanning in $F$, or hit in $F$.
We observe that Algorithm~\ref{alg:new_strategy} only considers occupied edges if all edges on $\OCH(F)$ are dividing,
in which case any occupied edge corresponds to a hit subchain.
This implies that we may execute Algorithm~\ref{alg:new_strategy} and so:

\begin{restatable}{theorem}{disks}
\label{thm:disks}
    Let $F$ be a family of $n$ pairwise disjoint disks with radii in $[1, k]$, or, a family of $n$ unit disks of ply $k$. 
    We can preprocess $F$ using $O(n \log n)$ space and $O(n ( k \log^3 n))$ time such that for $P \sim F$, we can reconstruct $\CH(P)$ in $O( r(F, P) \cdot k \log^3 n)$  time. 
\end{restatable}

\bibliography{refs}

\appendix

\newpage
\section{Related work}
\label{app:related}

We elaborate on the two related works most closely aligned with this paper.
Bruce, Hoffmann, Krizanc, and Raman~\cite{bruce2005efficient} presented the first instance-optimal retrieval strategy. However, they did not provide an explicit retrieval program with an accompanying algorithmic running time. We elaborate on their retrieval strategy and discuss the complexity considerations that arise when implementing a corresponding retrieval program.

Recently and independently of our work, L\"{o}ffler and Raichel~\cite{loffler2025preprocessing} developed an algorithm for reconstructing the convex hull when $F$ is a set of unit disks. They restrict the overlap in $F$ to ply $k$. The number of retrievals performed by their algorithm is not instance-optimal. We discuss their approach and provide commentary on its running time.

\subparagraph{The algorithm by Bruce, Hoffmann, Krizanc, and Raman~\cite{bruce2005efficient} and its running time.}
Bruce et al.~\cite{bruce2005efficient} present an instance-optimal retrieval strategy based on categorising regions in $F$. The key idea is to partition the curent set of regions into four categories (see Figure~\ref{fig:partition}). Each region is assigned the highest applicable category, as follows:

\begin{itemize}
    \item \emph{Always}: A region $R \in F$ is \emph{always} if for all $p \in R$ and all $P' \sim (F - R)$, the point $p$ lies on the convex hull of $P' \cup \{p\}$.
    \item \emph{Partly}: A region $R \in F$ is \emph{partly} if there exists $p \in R$ such that for all $P' \sim (F - R)$, the point $p$ lies on the convex hull of $P' \cup \{p\}$.
    \item \emph{Dependent}: A region $R \in F$ is \emph{dependent} if there exists $p \in R$ and $P' \sim (F - R)$ such that $p$ lies on the convex hull of $P' \cup \{p\}$.
    \item \emph{Never}: A region $R \in F$ is \emph{never} if for all $p \in R$ and all $P' \sim (F - R)$, the point $p$ never lies on the convex hull of $P' \cup \{p\}$.
\end{itemize}

\noindent
A region is classified as \emph{partly} if and only if it is not \emph{always} and it appears on the outer convex hull of $F$. Let $F$ denote the current set of regions, where some regions $R_i$ have already been retrieved and replaced by their corresponding point $p_i$. The retrieval strategy then:

\begin{itemize}
    \item If there exists a region $R \in F$ that is \emph{partly}, then they show the existence of a constant-size \emph{witness} set $W$ such that $R \in W$. All regions in $W$ are retrieved, and the algorithm recurses.
    \item Otherwise, if there exists a region $R \in F$ that is \emph{dependent}, then -- by the absence of partly regions -- they again show the existence of a constant-size witness $W$ with $R \in W$. The regions in $W$ are retrieved, and the algorithm recurses.
    \item If no region is partly or dependent, then it follows that for all $P_1 \sim F$ and $P_2 \sim F$, the cyclic order $\precorder(\CH(P_1)) = \precorder(\CH(P_2))$, and the algorithm terminates.
\end{itemize}

While the strategy is instance-optimal, no explicit retrieval program or running time analysis is provided. Constructing such a program would be a substantial contribution in its own right. Nevertheless, we make two observations that any such retrieval program would need to address.
First, the algorithm requires maintaining, at all times, whether any partly or dependent region exists. A natural approach is to dynamically maintain two sets: $A \subset F$, the set of partly regions, and $B \subset F$, the set of dependent regions. However, these sets can exhibit $\Theta(n)$ \emph{recourse}: after retrieving a single region, up to $\Theta(n)$ other regions may appear on the outer convex hull and become partly. It remains unclear whether these sets can be maintained in amortised constant time.
Second, the program must compute, for each selected partly or dependent region $R$, a corresponding \emph{witness} set $W$. The definition of a witness can be found in~\cite{bruce2005efficient}. We note that the witness $W$ of a region $R$ cannot be precomputed: the set of regions that may be a witness for $R$ can change dramatically after a retrieval.

\begin{figure}[t]
    \centering
    \includegraphics[]{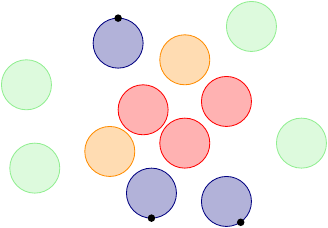}
    \caption{A set of (point) regions  $F$ and a partition of these regions into always (green), partly (blue), dependent (orange) and never (red) regions. For partly regions $R$, we illustrate a point $p \in R$ where $\forall P' \sim (F - R)$, $p$ lies on the convex hull of $P' \cup \{ p \}$. }
    \label{fig:partition}
\end{figure}

\subparagraph{The algorithm by L\"{o}ffler and Raichel~\cite{loffler2025preprocessing}.}
L\"{o}ffler and Raichel~\cite{loffler2025preprocessing} consider the original input set $F$ and partition it into three categories: always regions, never regions, and a third group consisting of partly and dependent regions. Their retrieval strategy is simple: they retrieve all regions in this third group.
Let $W(F) = \max_{P \sim F} r(F, P)$ denote the worst-case number of retrievals over all realisations $P \sim F$. Their algorithm retrieves $W(F)$ regions. As such, their result is not instance-optimal but \emph{worst-case optimal given $F$} -- a strictly weaker guarantee. This makes it easier to achieve, since the set of $W(F)$ regions can be identified during preprocessing.

There exist sets $F$ for which $W(F) \ll n$, and the strength of their contribution lies in designing a retrieval program whose running time depends on $W(F)$ rather than $n$. Specifically, when $F$ is a set of unit disks of ply $k$, their reconstruction algorithm runs in $O(k^3 W(F))$ time.
Thus, just as in this work, there exist sets $F$ for which their reconstruction takes sublinear time.
However, we also note that there exists sets $F$ and point sets $P \sim F$ where their reconstruction takes linear time, and ours constant.

\end{document}